\UseRawInputEncoding
\documentclass[final]{svjour3}    

\usepackage{lineno,hyperref}

\usepackage{times}
\usepackage{helvet}
\usepackage{courier}
\usepackage{amssymb}
\usepackage{amsmath}
\usepackage[usenames]{color}
\usepackage{algorithm}
\usepackage{algorithmicx}
\usepackage{algpseudocode}
\usepackage{graphicx}
\usepackage[caption=false]{subfig}

\usepackage{natbib}

\newtheorem{cor}{Corollary}
\newtheorem{conj}{Conjecture}

\newcommand{\veca}{\mathbf{a}}
\newcommand{\vecb}{\mathbf{b}}
\newcommand{\vecs}{\mathbf{s}}
\newcommand{\topc}{\mathbf{top}}
\newcommand{\calF}{\mathcal{F}}

\newcommand{\score}{\mathbf{sc}}

\newcommand{\default}{\psi}
\newcommand{\pord}[3]{#1 \succ_{#3} #2}

\usepackage{array}
\newcolumntype{L}{>{\centering\arraybackslash}m{1.5cm}}

\begin{document}

\title{Reaching Consensus Under a Deadline}

 


 \author{Marina Bannikova \and Lihi Dery \and Svetlana~Obraztsova \and Zinovi Rabinovich \and Jeffrey S. Rosenschein }

\institute{
       Marina Bannikova \at 
        Universitat Aut\`onoma de Barcelona, Spain
        \email {marina.bannikova@uab.cat}
        \and
        Lihi Dery (corresponding author) \at
              Ariel University, Israel 
              \email{lihid@ariel.ac.il}           
        \and
          Svetlana Obraztsova \at
              Nanyang Technological University, Singapore
              \email {lana@ntu.edu.sg}
        \and
        Zinovi Rabinovich \at
        Nanyang Technological University, Singapore
              \email {zinovi@ntu.edu.sg}
        \and
        Jeffrey S. Rosenschein \at
        The Hebrew University of Jerusalem, Israel
        \email{jeff@cs.huji.ac.il}
}
 
\date{Received: March 2020 / Accepted: December 2020}
\journalname {Autonomous Agents and Multi-Agent Systems}

\maketitle

\begin{abstract}


Group decisions are often complicated by a deadline. For example, in committee hiring decisions the deadline might be the next start of a budget, or the beginning of a semester. It may be that if no candidate is supported by a strong majority, the default is to hire no one - an option that may cost dearly. As a result, committee members might prefer to agree on a reasonable, if not necessarily the best, candidate, to avoid unfilled positions.
In this paper we propose a model for the above scenario---\emph{Consensus Under a Deadline (CUD)}---based on a time-bounded iterative voting process. We provide convergence guarantees and an analysis of the quality of the final decision. An extensive experimental study demonstrates more subtle features of CUDs, e.g., the difference between two simple types of committee member behavior, lazy vs.~proactive voters. Finally, a user study examines the differences between the behavior of rational voting bots and real voters, concluding that it may often be best to have bots play on the voters' behalf. 


\end{abstract}

\section{Introduction}\label{sec:intro}
We study the problem of arriving at a joint decision (a consensus) under a deadline, based on the preferences of multiple voters. The voters' task is to find an alternative that the majority agrees upon, before some predefined deadline is reached. The majority is predefined and can vary from 51\% of the votes to unanimous agreement. 
As different individuals have different preferences, a decision does not occur immediately, and more than one round of voting may be required. Hence, a voting process takes place, where voters potentially change their ballots as the deadline approaches.

We define a strict, formal, time-bounded iterative voting process. The process begins with each voter revealing her most-preferred alternative. 
If no alternative reaches the required majority, a multi-stage voting process begins.
At each stage, all voters that wish to change their ballot apply for a voting slot. A voter may choose to change her ballot if, for instance, she realizes that her most-preferred alternative has no chance of being elected. She then might decide to vote for another alternative. One voter is chosen randomly from all those who applied for a voting slot; the chosen voter casts her new ballot. Then, the voting result is updated and publicized. The process continues in rounds until a consensus is reached, or until the deadline is reached, the sooner of the two. 
We assume that each voter has a private, strict preference order over the alternatives and that the number of alternatives is fixed. Bargaining is not permitted. We further assume that the voters are rational but strategic. The amount of strategic ballot changes is unlimited and subject only to the deadline constraint. Each stage (or round) is defined as one clock tick. Lastly, we assume that a failure to reach a consensus is the worst outcome for all voters. 

\subsection{Motivation} 
As a motivating scenario, consider an academic hiring committee. There are a few shortlisted candidates, and each committee member has reviewed their merits, attended their talks, and formed a preference order over the candidates, reflecting her opinion of them.
The committee has to reach a decision with a large majority, and has only a limited time to decide before the next budget period or academic term begins. If the committee fails, the lack of a new faculty member might have a budgetary impact, and cause a higher teaching load - the worst outcome for everyone. The faculty's secretary provides support for the ensuing iterative voting procedure. He collects votes, summarizes them and publicizes the results to all committee members.
However, the secretary can perform only a finite, reasonable number of such steps.
So if any committee member wishes to change her vote, there is only a finite, predetermined number of opportunities she will have to do so.
Under our model, the secretary will only accept a vote change from one committee member per iteration.

The restriction of one vote change per round is not an arbitrary design choice, but often matches real-world needs. As a practical example, consider an event that occurred in the hometown of one of the authors. The town is divided into two neighborhoods; each of the neighborhoods contains one public elementary school. 
According to Ministry of Education regulations, a class cannot contain more than 42 children. Furthermore, a class may only be split into two classes, if it contains more than 42 children. 
Children entering first grade are automatically assigned to the school in their neighborhood. 
However, recently, 41 first grade children were allocated to one neighborhood school and 30 children to the other. 
In an attempt to decrease class sizes, and for the benefit of the young children, the town canceled the automatic allocation and allowed parents to register their child to either one of the schools, hoping that this would lead to either three small-sized first grade classes, or at least two more equally-sized classes.
As both schools are perceived as equally good, parents prefer their children to attend the school in their own neighborhood. However, they also prefer their children to study in a class that contains a minimal number of children.
Sadly, no registration requests were filed. Parents assumed other parents would change their registration and were reluctant to do so themselves, even though in hindsight, each parent regretted not changing their registration (a no-action variant on ``tragedy of the commons''). 
Had the town performed an iterative voting process, publicizing the current results at each stage, and bounded by a deadline, these first graders would likely have begun their education in less-crowded classes.

Other examples, where consensus under a deadline is required, include 
a selection of a CEO of a company, a scientific committee deciding where to hold next year's conference, and even more informal problems, such as a major holiday family dinner venue.

All the above cases have several common features. First, there is a strict deadline for reaching an agreement: the start of an academic year, national holidays, the budget approval deadline, or dinnertime.
Second, assuming that individuals at least somewhat differ in their preferences, it is unlikely that a consensus will be reached immediately; a consensus is usually reached, if at all, only after several rounds of a sequential voting process. 

\subsection{Contributions}
We provide, to the best of our knowledge, the first model for iterative voting with a deadline (Section~\ref{sec:model}). The model can be generalized to other voting rules, but for ease of exposition, we initiate this line of research with a specific model, namely, \emph{Consensus Under a Deadline} (CUD), based on Plurality with a threshold (also known as Majority). We establish the theoretical properties of CUDs, such as termination, guarantees of no runs ending with a default, and additive Price of Anarchy bounds (Section~\ref{sec:cud_theory}). CUDs adhere to a simple protocol, and we can effectively simulate them to investigate statistical properties; we provide an encompassing experimental analysis of their tradeoffs (Section~\ref{sec:simulations}). 
In particular, we measure the effects of voter behavioral types (lazy vs.~proactive) on the number of voting steps and the additive Price of Anarchy. 

Furthermore, we introduce the CUD game, an on-line voting game that we designed as a user-study, in order to examine how real users behave (Section~\ref{sec:userstudy}). The collected data allows us to analyze human behavior such as how often a consensus is reached, and when (if at all) strategic voting occurs; perhaps most importantly, we also examine the quality of the final decision reached by a group of people, in terms of the average reward and the additive price of anarchy. We compared our results with the performance of rational bots, and a mixture of bots and humans. Consequently, we provide some insight into bot vs.~human voting and propose that bots may have a big advantage in consensus reaching scenarios with a deadline, as they increase the convergence percentage and the general satisfaction of all voters. 


A short version of this paper, with preliminary results, was presented in a workshop \citep{bannikova16}. That version contained theorems with no proofs and no examples, and some experimental simulations.
The CUD game was presented as a short conference paper \citep{Yosef:2017}. 
An independent research that used the CUD game to study Alexithymia traits was led by one of the authors ~\citep{Gvirts_Dery}. 


\section{Related Work}\label{sec:rel_work}
Reaching a consensus has been extensively studied in the field of group decision making \citep{herrera2007consensus, DONG201445,zhang2019consensus}. Recently, social influence has began to be explored \citep{WU201739, capuano2017fuzzy}. However, we study a scenario where voters do not interact or bargain. 

Similarly to us, Yager and Alajlan \citeyearpar{YAGER2015170} suggest iterative rounds where voters provide modified versions of their subjective probability distribution. In each round the voters see the current group aggregated probability distribution and are allowed to change their own probability distribution. However, differently from us, in their model, the voters  need to agree on the \textit{aggregated distribution}. In our model, voters are all required to agree on the same candidate, i.e, they are required to submit the same distribution.

In a Consensus Under a Deadline (CUD) game, a \emph{consensus} on one \emph{alternative} must be reached within a fixed \emph{deadline}. If the voters do not reach a consensus, a default candidate is declared. 

To the best of our knowledge, this paper is the first to analyze theoretical features of CUD games in iterative voting, and to demonstrate them experimentally, with a focus on convergence time and final outcome quality. Furthermore, we study human behavior in CUDs, our goal being to examine whether humans behave rationally.

\subsection{Myopic voters}
In principle, hierarchies of beliefs and cognitive hierarchies, have been around for quite some time and across the AI board (e.g., definitely non-exhaustively, ~\cite{Halpern:1985:GML:1625135.1625229,monderer_samet_89,Halpern:1990:KCK:79147.79161,boergers_94}), and even of a freshly renewed interest in voting games specifically (see, e.g., \cite{chc_2004,cc-gi_2013,elkind2020cognitive}). 
In reality, capabilities of the majority of people are limited to very shallow cognitive hierarchies even in simple games, as shown in experiments (see \cite{stahl_wilson_1994, nagel_95}).
Handling both beliefs about other player behaviors and their preferences, is a much more difficult proposition even for computational agents (see, e.g., \cite{gmytrasiewicz_doshi_2005}).
        \footnote{Many a joke deals with chains of ``I know that he thinks that I know...", and are in fact based on human inability to handle deep nested beliefs. Comically this is even witnessed in popular culture, e.g., the episode of ``Friends'': ``The one where everybody finds out''.} 
Furthermore, some of the more recent analysis suggests that a large proportion of people are in fact myopic \citep{wl-b_aaai_2010}. 
Therefore, because we seek a model that would easier relate to the real-world human behavior trends, we adopt the assumption of myopic agent strategies.

\subsection{Iterative Voting Games and the question of convergence}
A CUD game is a type of iterative voting game (see, e.g.,~\cite{MPRJ10,reijngoud2012voter,n-dork_2015,omprj_2015_aaai, meir2016strong, dery2019lie}). However, CUDs have several unique  features.
First, although CUDs do utilise a known voting rule, they work directly with the set of possible winners (i.e., alternatives that might be chosen by the majority), and behave much like non-myopic games based on local-dominance (see, e.g.,~\cite{Meir:2014:EC,loprr_AGT_2015,m_2015}).
On the other hand, the distinction between  lazy and proactive voter behaviour links CUDs with biased voting (see, e.g.,~\cite{emos_2015}).

The standard assumption in the iterative voting literature is that voting processes do not terminate, unless no voter has a way to further manipulate the outcome. Convergence is the subject of an extensive research effort, both to determine when these processes terminate and with what ballot profile (\cite{MPRJ10,reijngoud2012voter,RW12,omprj_2015_aaai,koolyk2016convergence}). 
CUDs, on the other hand, always terminate and declare a winner (see Theorem~\ref{thm:cuds_stop}). Our framework will always converge if the voters are rational and there is a sufficient amount of time before the deadline.

Some models focus on choosing one alternative with a certain degree of consensus. The consensus can be a simple majority, a majority of $2/3$ and so forth. For example, a Papal Conclave has to choose the next pope with a majority of at least $2/3$. The voting process continues infinitely until the required majority is reached. Here, the voters' degree of indifference to time plays a major role. 
For example, Kweik \citeyearpar{Kwiek2014} has shown that when all voter preferences are known in advance,  the outcome may be predicted, and it favors the voter with the highest indifference to time.

Reijngoud et al. \citeyearpar{reijngoud2012voter} analyzed the iterated voting model from different angles: how the amount of information received in one single poll influences the manipulated behavior of the voter, and how repeated polls affect the manipulated behavior of all voters.
Br{\^a}nzei et al. \citeyearpar{branzei2013bad} studied the different influence of strategic voting on the main voting methods. They defined the \textit{additive Price Of Anarchy} ($PoA^+$) factor that measures the differences between truthful voting and worse-case manipulated voting. Though the general concept of the price of anarchy historically had other interpretations and uses (see, e.g., \cite{koutsoupias1999worst,papadimitriou2001algorithms}), it is the additive form that best befits our setting. Thus, $PoA^+$ is one of the measures we use in our experimental evaluation.


\subsection{Deadline}
A feature that distinguishes our framework from other iterative voting processes is the \emph{deadline}. Informally, the deadline is a time limit on the voting process. Bargaining models often have a deadline; however, they do not usually require a consensus or the selection of one alternative only. Rather, the outcome may be a compromise among the most-preferred outcomes by all participants. Deadlines are defined as ``fixed time limits that end a negotiation''~\citep{Ma1993,Moore2004}. Experimental studies of bargaining deadlines have shown that the majority of agreements are obtained in the final seconds before the deadline \citep{FERSHTMAN1993306,KOCHER2006375}. Even complete information does not speed agreement much, contrary to what might be expected~\citep{Roth2016}; moreover, agents are convinced that revealing their deadlines may decrease their payoffs, contrary to what is proven theoretically \citep{MOORE2004125}. 
The concession rate increases as the deadline approaches \citep{Cramton1992,Lim1994,KOCHER2006375}. If the deadline is soon enough, then the agents converge almost immediately, but if the deadline is long enough away, there probably appears a delayed equilibrium \citep{Gale1995}. 
 If the deadlines are different for the agents and are common knowledge for everyone, then the theory predicts an inmediate equilibrium \citep{Kwiek2014}, yet some experiments show that there will be delays \citep{Hurkens2015}.
However, there is a  probability that no agreement will be reached before the deadline; depending on the context, more than half \citep{Ma1993} or at least one third \citep{MERLO2004192} of all negotiations may fail. Nonetheless, a moderate deadline has a positive effect on the outcome~\citep{Moore2004}.
Our model does not encompass negotiation or bargaining.


 \subsection{Human Behavior in Voting Games}
On-line platforms such as Doodle allow people to submit their votes on the problem in question, to view the votes of other group members, and to change their votes if they wish.
Zou et al. \citeyearpar{zou2015strategic} studied strategic voting behavior in Doodle polls, and how the knowledge of current voters' preferences affected results. Interestingly, people tend to change their vote preferences in light of the votes of their peers. In our model,  voter preferences are private, and voters can only see intermediate aggregated results.

\cite{meir2020strategic} performed a study of human behavior in online voting. Similarily to us, their setting is of an iterative voting game, under the plurality voting rule.
They classified voters into three distinct types, two of which are not strategic and one that will perform straightforward strategic moves. However their setting does not contain a deadline. Moreover, while they provide valuable insights on human voting patterns, they do not compare their framework with a rational strategic model, so the question of whether humans and bots act alike remained unanswered.
 
\section{Model}\label{sec:model}

In this section we formally model a deadline-bounded, iterative voting process.
We concentrate on an extension of the Majority rule, i.e., on plurality voting with a threshold, leaving other voting rules for future work. 
The threshold determines the number of voters that need to vote for the same candidate for that candidate to win. The threshold is termed {\em tight} if it is equal to the number of voters, i.e., a unanimous consensus vote is required, and is termed {\em relaxed} otherwise.

The process begins when each voter reveals her most preferred alternative. The preferences are aggregated using plurality voting. If a consensus is not instantly reached, a process of vote alterations begins. 
Formally, let $V$ be a set of $n$ voters choosing 
from a set, $C$, of $m+1$ alternatives in an iterative process with a \emph{deadline} of $\tau$ steps. At any intermediate stage of the voting process, $t\in[0:\tau]$ will denote the remaining number of steps before the process must stop.
We denote by $\default\in C$ the default alternative, and $C^+=C\setminus\{\default\}$ the set of \emph{valid alternatives}. Notice that $|C^+|=m$.
Each voter is characterised by a \emph{truthful preference} $a_i\in\mathcal{L}(C)$, where $\mathcal{L}(C)$ is the space of complete and non-reflexive orderings over the set of alternatives $C$. Generally, we write $a_i(c,c')$ if voter $i$ prefers $c$ to $c'$. However, when it is beneficial to visually underline the relative preference of two options, we will also resort to the notation $c\succ_i c'$. We assume that $a_i(c,\default)$ for all $i\in V$ and $c\in C^+$, i.e., the default alternative is the worst option from the point of view of all voters. 
At the beginning of the process, when $t=\tau$, every voter $i\in V$ casts a ballot $b_i^t \in C$ that reveals her most preferred candidate.
The collection of all ballots, $\vecb^t=(b_1^t,\ldots,b_n^t)\in C^n$, is a \emph{ballot profile}.
Given a ballot profile, the \emph{score} of each candidate $c$ is the sum of votes for this candidate:
$\score_c(\vecb^t)=|\{i\in V | b_i^t=c\}|$. Please notice that individual candidate scores are scalars, i.e., $\score_c(\vecb^t)\in\mathcal{N}$.

A \emph{score vector} is a collection of scores of all \emph{valid} candidates
$ \score(\vecb^t)=$ \\ $(\score_{c_1}(\vecb^t),\ldots,\score_{c_m}(\vecb^t))\in\mathcal{N}^m. $
For convenience, a shorthand $\vecs^t=(s_1^t,\ldots,s_m^t) = \score(\vecb^t)$ is used, and we omit the time superscript to denote an arbitrary score vector.
The score vector is public knowledge to all voters.

Possible outcomes at time $t$ are computed using a \emph{Multi-stage Defaulted Voting Rule (MDVR)}, that estimates whether any non-default candidates have enough time to gather sufficient support to cross the majority threshold. MDVR explicitly focuses on iterative ballot modifications, and accounts for the possibility that the score of a candidate may change. 
Formally,
an MDVR, denoted $\calF$, maps a score vector at time $t$ to a subset of all candidates (including the default option), and has the form $\calF:\Delta\times[0:\tau]\rightarrow 2^C$, where $\Delta$ is the space of all feasible score vectors.

In this paper we investigate two MDVRs: the \emph{Iterative Majority (IMaj)}, and its special sub-case, \emph{Iterative Unanimity (IUn)}.
In the following we will assume that a majority threshold is $\sigma>\frac{n}{2}$, and naturally extend this to the unanimous threshold of $\sigma = n$.
Let $\widehat{W}\subseteq C^+$ denote the set of all valid candidates whose 
score difference from the winning majority threshold is bounded by the time until the deadline. That is:
$\widehat{W}(\vecs,t)=\left\{c\in C^+\ |\ \sigma-s_c<t+1\right\}$. Intuitively, these are the candidates that may still gather winning support if enough voters begin to favour them within the remaining time until the deadline. Throughout the paper we will term $\widehat{W}=\widehat{W}(\vecs,t)$ \emph{the set of possible winners}.\footnote{Note the difference from the classical concept of possible winners, which refers to expansions of partial preference ballots (see e.g., \citep{xia_conitzer_2011}).} Notice that the time index $t$ here means that this set is calculated {\em after} voters have expressed their vote alterations at time-slice $t$.
Thus the rule of IMaj is defined as:
$$\calF^{IMaj}_\sigma(\vecs,t)= \begin{cases}\{\default\}& if~ \widehat{W}=\emptyset\\\widehat{W}&otherwise\end{cases}
$$

Once the deadline is reached, $\calF(\vecs,0)$ is a singleton containing either (i) a valid candidate with the highest score ({\bf the winner}), or (ii) the default candidate $\default$ ({\bf default}).
Notice that the MDVR 
for Iterative Unanimity is $\calF^{IUn}=\calF^{IMaj}_n$.

Prior to the deadline, however, the set of possible winners is subject to the manipulative behaviour by voters. Specifically, at any time point $t\in[0:\tau]$ prior to the deadline, having calculated and observed the set of possible winners, each voter considers how a change in her ballot may influence the set calculation at the next time point. I.e, each voter decides whether to change her vote and produces a new ballot $b_i^t\in C$, thus (re)stating her vote at time $t$. Assuming, as we do here, that all voters are selfish and myopic leads to a particular pattern of behaviour w.r.t this ballot selection. Namely, if at time $t$ the set of possible outcomes includes more than one valid candidate, a voter will choose her ballot to support the more favoured candidates over the less preferred ones. In fact, we presume that a utility function can be defined to express the beneficial effects of a ballot alteration on the set of possible winners, and that all voters myopically seek a ballot that maximises this utility. We formalise this utility point of view as follows.


{\bf The voter's best possible outcome:} For any non-empty $W\subset C$ the best alternative in $W$ w.r.t.~the voter's truthful preferences $a_i$ is denoted $\topc_i(W)\in W$. That is, $w$ is the best possible outcome for voter $i$ if voter $i$ prefers it over any other candidate in the possible outcome set: $w=\topc_i(W)$ if and only if for all $c\in W\setminus\{w\}$ holds $a_i(w,c)$.

{\bf The voter's utility function:} Each voter has a utility function that matches a score vector at time $t$ to the voter's utility: $u_i:\Delta\times[0:\tau]\rightarrow\mathbb{R}$. We assume that the utility function is consistent with the voter's truthful preferences $a_i$ for any $t \in [0: \tau]$. We will particularly emphasise the classes of \emph{lazy consistent} and \emph{proactive consistent} utility functions. Intuitively, a lazy consistent utility function drives the voter to change her ballot only if this creates a set of possible outcomes, while a proactive consistent function would also induce ballot change if the score of the best possible outcome can be improved. Formally, these utility classes are defined as follows.

\begin{definition}\label{def:lazy}
A utility function $u_i$ is \emph{lazy consistent} if for all $\vecs,\vecs'\in\Delta$ and $t\in[0:\tau]$ the following condition holds:
$$u_i(\vecs,t)>u_i(\vecs',t)\iff a_i(w,w'),$$
where $w = \topc_i(\calF(\vecs,t))$ and $w'=\topc_i(\calF(\vecs',t))$.
\end{definition}

\begin{definition}\label{def:act}
A utility function $u_i$ is \emph{proactive consistent} if for all $\vecs,\vecs'\in\Delta$ and $t\in[0:\tau]$  holds the condition that:
$$u_i(\vecs,t)>u_i(\vecs',t)\iff a_i(w,w')\lor \left( (w=w')\land (s_w>s'_w)\right)$$
where $w = \topc_i(\calF(\vecs,t))$ and $w'=\topc_i(\calF(\vecs',t))$.
\end{definition}

We further assume that all utility functions are homogeneously either lazy consistent or proactive consistent. Notice that, although here we only use $\calF^{IMaj}$ and $\calF^{IUn}$, CUDs can be naturally extended to more general forms of MDVRs. It is also important to note that the use of cardinal utilities has some formal benefits. Although it is possible to describe both lazy and proactive voter behaviour using only ordinal preferences, that leads to formal descriptions of choosing the best course of actions, or of reasoning about preference alterations, becoming vividly distinct for the two behaviours. The use of utilities allows us to avoid this complication, creating an elegant, uniform formal treatment of our model under various voter behaviour characteristics. 

We must underline, however, that we do not assume knowledge of the voter's utility function or its properties beyond consistency. In fact, we treat the voter's utility as a subjective valuation that is private to the voter. The only use for the voter's utility we have is the convenience of formal presentation.

Now, it must be noted that for $\calF^{IMaj}$ and $\calF^{IUn}$ the definitions of lazy and proactive utilities do not differ in their behaviour if $\calF(\vecs,t)=\{\default\}$. That is, once MDVR provides the default as the only possible outcome, neither lazy nor proactive behaviour dictates any change. Therefore, we concentrate our discussion on the effect of these behaviours in those cases where $\calF(\vecs,t)\subseteq C^+$, that is the rule produces a set of {\em possible winners} $\widehat{W}\neq\emptyset$ as its set of possible outcomes. In fact, the proof of Theorem~\ref{thm:cuds_stop} will later show that strategic considerations by voters can only be effective in those cases where $\calF(\vecs,t)\subseteq C^+$ and a possible winner set $\widehat{W}$ is involved.

The following example illustrates the voting procedure.

\begin{example}Example of the voting procedure. 
\label{example:voting_model} 

Let there be three voters, $V=\{1,2,3\}$ and three alternatives $C^+=\{a,b,c\}$. Assume the strictest voting rule, unanimity, i.e., $\sigma=n$. The time until the deadline is $\tau=2$; the voters are required to reach a consensus in three time periods, from time $t=\tau$ until $t=0$. We assume that the three voters have the following preferences:


voter 1: $a \succ_1 b \succ_1 c$;

voter 2:  $b \succ_2 c \succ_2 a$;

voter 3:  $c \succ_3 a \succ_3 b$.

 \textbf{Stage $t=2$.} The voters cast their ballots and reveal their (truthful) most preferred alternative, $b{^\tau}=(b_1^{\tau},b_2^{\tau},b_3^{\tau})=(a,b,c)$. Thus the scores are $s^2=(s^2_a,s^2_b,s^2_c)=(1,1,1)$.

In order to win, an alternative must receive all $\sigma=3$ votes. At the moment each alternative has only one vote. There are two stages ahead, hence any alternative can gather the two remaining votes. Therefore at this stage all the three alternatives $\{a,b,c\}$ are possible winners, $\widehat{W}(s^2,2)=\{a,b,c\}$.

\textbf{Stage $t=1$.} Each voter decides whether she wishes to change her vote.
Note that the voters consider their options before they know if they will be chosen to cast their ballots. This can be viewed as answering the question: ``what will I vote if I am picked to vote?''. Each voter asks herself this question and decides whether she wishes to have a possibility to change her ballot. In the current example, the voters know that someone must change the ballot, or else the default alternative will be chosen (and this default is the worst outcome for all voters). The voter does not know if the other voters will agree to change their ballot, and since she wishes to avoid the default alternative from being chosen, she will agree to change her ballot.

For instance, consider voter 1. She has 3 possible actions: either she stays with $a$ ($b_1^1=a$), or she changes to $b$ or $c$ ($b_1^1=b$ or $b_1^1=c$). If she declares that she wishes to change, she might be randomly picked to do so. In this case, if she is picked to change at the current stage, she can predict the consequences of each of these possible votes, as Table \ref{tab:v1tau} illustrates.

If she decides to keep her vote for $a$ (row 2 of Table \ref{tab:v1tau}), then the scores of the alternatives will be $sc(b^1)=(s_a^1,s_b^1,s_c^1)=(1,1,1)$ and there will be no possible winner, because  $n-s_x^1\geq\tau-1+1$ for any alternative $x \in C$.

If voter 1 decides to change her ballot, and vote $b$ instead of $a$ (row 3 of Table \ref{tab:v1tau}),  the scores are $sc(b^1)=(s_a^1,s_b^1,s_c^1)=(0,2,1)$. Since $n-s_b\leq\tau-1+1$, alternative $b$ is a possible winner.

 \begin{table}[H]\caption{The possible decisions of voters 1, 2 and 3 at stage $t=1$. Each voter can vote for either $a$, $b$ or $c$ (column 2). The corresponding scores for these choices appear in column 3. The possible winners following these choices appear in column 4.}
 \label{tab:v1tau}
 \begin{center}

 \begin{tabular}{|c |c |c| c| c|}\hline
 \rule{0cm}{0.5cm}
 Voters &$b^1_i$ & $(s_a^1,s_b^1,s_c^1)$ & $\widehat{W}(s^1,1)$&\\
 \hline
 \rule{0cm}{0.5cm}
 voter 1& $a$ & (1,1,1) & $\{\}$&\\
 & $b$ & (0,2,1) & $\{b\}$& $\leftarrow$  prefers to switch to $b$\\
 & $c$ & (0,1,2) & $\{c\}$&\\
 \hline
 voter 2& $a$ & (2,0,1) & $\{a\}$&\\
 & $b$ & (1,1,1) & $\{\}$&\\
 & $c$ & (1,0,2) & $\{c\}$&$\leftarrow$   prefers to switch to $c$\\
 \hline
 voter 3& $a$ & (2,1,0) & $\{a\}$& $\leftarrow$  prefers to switch to $a$\\
 & $b$ & (1,2,0) & $\{b\}$&\\
 & $c$ & (1,1,1) & $\{\}$&\\
 \hline
 \end{tabular}
 \end{center}
 \end{table}

If at this stage voter 1 decides to switch to $c$ (row 4 of Table \ref{tab:v1tau}), then the scores are $(s_a^{\tau-1},s_b^{\tau-1},s_c^{\tau-1})=(0,1,2)$, and $c$ would be a possible winner. Since the utility function is consistent with the truthful preference, voter 1 prefers to switch to $b$. This is true both for lazy consistent and proactive voters.

The same argument can be applied to voters 2 and 3. Therefore, all three voters ``raise their hands'' and wish to change their votes. Possible changes are gathered in Table~\ref{tab:v1tau}.

We assume that voter 1 is randomly picked to change her vote. Accordingly, $b^1=(b^1_1,b_2^1,b_3^1)=(b,b,c)$ and $s^1=(s_a^1,s_b^1,s_c^1)=(0,2,1)$, which makes $b$ the only possible winner.

\textbf{Stage $t=0$.}
Voters 1 and 2 have no reason to change their vote since, as is shown in Table \ref{tab:v12t0}, whatever they vote, they cannot be better off. Voter 3 can be better off by switching to $b$, as is illustrated in Table \ref{tab:v12t0} since, if she votes for $b$, the ballot would be $b^0=(b_1^0,b_2^0,b_3^0)=(b,b,b)$ and the scores would be $s^0=(s_a^0,s_b^0,s_c^0)=(0,3,0)$, which makes $b$ the only possible winner. Consequently, voter 3 decides to change her vote from $c$ to $b$. Since she is the only one who wishes to change (that is, she is the only possible voter to be randomly picked), she changes to $b$.

 \begin{table}[H]\caption{Possible decisions of voters 1,2 and 3 at stage $t=0$. Voters 1 and 2 have no reason to change their vote. Voter 3 will prefer to change to $b$. }\label{tab:v12t0}
 \begin{center}
 \begin{tabular}{|c| c| c| c| l|}
 \hline
 \rule{0cm}{0.5cm}
 Voters &$b^0_i$ & $(s_a^0,s_b^0,s_c^0)$ & $\widehat{W}(s^0,0)$&\\
 \hline
 voter 1& $a$ & (1,1,1) & $\{\}$&\\
 & $b$ & (0,2,1) & $\{\}$&\\
 & $c$ & (0,1,2) & $\{\}$&\\
 \hline
 voter 2& $a$ & (1,1,1) & $\{\}$&\\
 & $b$ & (0,2,1) & $\{\}$&\\
 & $c$ & (0,1,2) & $\{\}$&\\
 \hline
 voter 3& $a$ & (1,2,0) & $\{\}$&\\
 & $b$ & (0,3,0) & $\{b\}$& $\leftarrow$  prefers to switch to $b$\\
 & $c$ & (0,2,1) & $\{\}$&\\
 \hline
 \end{tabular}
 \end{center}
 \end{table}

After voter 3 changes her vote, the voting ballot is $b^1=(b,b,b)$ and $s^1=(0,3,0)$. Since $s_b=\sigma=n$, alternative $b$ is chosen as the winner.

\end{example} 

The behaviours of \emph{lazy} and \emph{proactive} voters are different with respect to the stage at which they decide to change, and to which alternative they switch. The difference is illustrated by the following two examples for both voter types.

\begin{example} Example illustrating the behavior of lazy voters.
\label{example:proactive_lazy}

There are five \emph{lazy} voters, $V=\{1,2,3,4,5\}$ and there are four alternatives to be voted for, $C^+=\{a,b,c,d\}$. Assume unanimity, $\sigma=n=5$ and that the voters need to reach a consensus in five time periods, $\tau=4$. We assume that voters have the following preferences:


voter 1: $a \succ_1 b \succ_1 c \succ_1 d$;

voter 2: $a \succ_2 c \succ_2 b \succ_2 d$

voter 3: $b \succ_3 c \succ_3 a \succ_3 d$;

voter 4: $b \succ_4 a \succ_4 c \succ_4 d$;

voter 5: $c \succ_5 b \succ_5 d \succ_5 a$

\textbf{Stage} $t=4$.
The voters reveal their truthfully most preferred alternative, hence, the voting ballot is $b^4=(b_1^4,b_2^4,b_3^4,b_4^4,b_5^4)=(a,a,b,b,c)$ and, therefore, the scores are $s^4=(s^4_a,s^4_b,s_c^4,s_d^4)=(2,2,1,0)$. The possible winners are $\widehat{W}(s^4,4)=\{a,b,c\}$.

\textbf{Stage} $t=3$.
Each voter needs to decide whether she wishes to change her vote. Voters 1--4 cannot be better off if they change their vote, since any change would exclude their top-choice from the set of possible winners. Therefore, they do not ``raise their hands'' in order to change their votes. This is illustrated in Table \ref{tab:v1-4t4}.

Consider voter 5. Whatever she votes, the set of possible winners consists only of $a$ and $b$ as is illustrated in Table~\ref{tab:v1-4t4}. Hence, by Definition \ref{def:lazy} voter 5, being a \emph{lazy voter}, prefers not to change her vote. 

 \begin{table} [H]\caption{Possible decisions of voters 1-5 at stage $t=3$. Voters 1-4 will not change their votes.}\label{tab:v1-4t4}
 \begin{center}
 \begin{tabular}{|c |c |c |c |l|}
 \hline
 \rule{0cm}{0.5cm}
 Voters &$b^3_i$ & $s^3$ & $\widehat{W}(s^3,3)$&\\
 \hline
 voter 1& $a$ & (2,2,1,0) & $\{a,b\}$& $\leftarrow$  prefers not to change\\
 & $b$ & (1,3,1,0) & $\{b\}$&\\
 & $c$ & (1,2,2,0) & $\{b,c\}$&\\
 & $d$ & (1,2,1,1) & $\{b\}$&\\
 \hline
 voter 2& $a$ & (2,2,1,0) & $\{a,b\}$& $\leftarrow$  prefers not to change\\
 & $b$ & (1,3,1,0) & $\{b\}$&\\
 & $c$ & (1,2,2,0) & $\{b,c\}$&\\
 & $d$ & (1,2,1,1) & $\{b\}$&\\
 \hline
 voter 3& $a$ & (3,1,1,0) & $\{a\}$&\\
 & $b$ & (2,2,1,0) & $\{a,b\}$& $\leftarrow$  prefers not to change\\
 & $c$ & (2,1,2,0) & $\{a,c\}$&\\
 & $d$ & (2,1,1,1) & $\{a\}$&\\
 \hline
 voter 4& $a$ & (3,1,1,0) & $\{a\}$&\\
 & $b$ & (2,2,1,0) & $\{a,b\}$& $\leftarrow$  prefers not to change\\
 & $c$ & (2,1,2,0) & $\{a,c\}$&\\
 & $d$ & (2,1,1,1) & $\{a\}$&\\
 \hline
 voter 5& $a$ & (3,2,0,0) & $\{a,b\}$&\\
 & $b$ & (2,3,0,0) & $\{a,b\}$&\\
 & $c$ & (2,2,1,0) & $\{a,b\}$& $\leftarrow$ prefers not to change\\
 & $d$ & (2,2,0,1) & $\{a,b\}$&\\
 \hline
 \end{tabular}
 \end{center}
 \end{table}

Hence, at stage $t=3$ the ballot is $b^3=(b_1^3,b_2^3,b_3^3,b_4^3,b_5^3)=(a,a,b,b,c)$, and alternatives' scores are $s^3=(s_a^3,s_b^3,s_c^3,s_d^3)=(2,2,1,0)$, which make alternatives $a$ and $b$ the only possible winners at the current stage.

\end{example}

\begin{example}\label{example:proactive_voters} Example illustrating the behavior of proactive voters.

Consider the same voters and preferences as in the previous example. When the voters are \emph{proactive} the result at \textbf{stage $t=4$} does not change; the possible winners are $\widehat{W}(s^4,4)=\{a,b,c\}$.

At \textbf{stage $t=3$} voters 1--4 cannot be better off if they change their vote, as is illustrated in Table~\ref{tab:v1-4t4}.

Consider voter 5. As  is illustrated in Table~\ref{tab:v5t4}, whatever she votes, the set of possible winners will be $\{a,b\}$. Since she is a \emph{proactive} voter, then, by Definition \ref{def:act}, given that she prefers $b$ to $a$, she ``raises her hand'' because she can be better off by voting for $b$. Being the only voter who wants to cast a ballot at the current stage, she is picked to change it.

 \begin{table} [H]\caption{Possible decisions of voter 5 at stage $t=3$. The voter will prefer to change her ballot.}\label{tab:v5t4}
 \begin{center}
 \begin{tabular}{|c |c |c |c |l |}
 \hline
 \rule{0cm}{0.5cm}
 Voters &$b^3_i$ & $s^3$ & $\widehat{W}(s^3,3)$&\\
 \hline
 voter 5& $a$ & (3,2,0,0) & $\{a,b\}$&\\
 & $b$ & (2,3,0,0) & $\{a,b\}$&$\leftarrow$ prefers to change\\
 & $c$ & (2,2,1,0) & $\{a,b\}$& \\
 & $d$ & (2,2,0,1) & $\{a,b\}$&\\
 \hline
 \end{tabular}
 \end{center}
 \end{table}

Hence, at stage $t=3$ the ballot is $b^3=(b_1^3,b_2^3,b_3^3,b_4^3,b_5^3)=(a,a,b,b,b)$, and alternatives have scores $s^3=(s_a^3,s_b^3,s_c^3,s_d^3)=(2,3,0,0)$, which makes alternatives $a$ and $b$ the only possible winners.

\end{example}

Having constructed some procedural intuition with Examples~\ref{example:voting_model} and~\ref{example:proactive_voters},
we can now summarise the progress of a CUD into a protocol.
Formally, a CUD iterative voting game proceeds as depicted in Protocol~\ref{cud_game}. Since score vectors simply accumulate ballots from $\vecb^t$, we use algebraic operations between a score vector and a ballot. That is, if $\vecs'=\vecs-c$ (respectively, $\vecs = \vecs+c$) then $s_k'=s_k$ for all $k\in C\setminus\{c\}$ and $s_c'=s_c-1$ (respectively, $s_c'=s_c+1$). In other words, ``adding'' a candidate to a score vector increases the score of that candidate by 1, while ``subtracting'' the candidate reduces it by 1.

 Protocol~\ref{cud_game} is parameterised by 
 the Multi-stage Defaulted Voting Rule 
 $\calF$ and 
 the deadline $\tau$. The game, therefore, proceeds as follows.
Iteratively, as long as the deadline has not been reached:
all voters calculate the current score vector (line 3).
If there is only one possible winner, $w$, a decision has been reached and the game ends (lines 4--5).
If the game continues, every voter calculates what is her best possible winner, given that the current score vector would be augmented by her vote alteration (line 7--8). If there are ties (i.e., if a few alternatives receive the same score), the voter selects the alternative that is ranked highest in her truthful preferences $a_i$.
Each voter decides whether she wants to change her vote, based on the possible winner set and her utility function. 
Voters who want to vote ``raise their hands'', i.e., are collected into a set $I$ (line 10).
A random voter is chosen from set $I$ (line 11).
The chosen voter casts her ballot (lines 12--15),
and the deadline is now one step closer (line 16).

Now, at first it may appear that Protocol~\ref{cud_game} is unstable, and has multiple issues. E.g., it may appear that the time can run out with more than one alternative remaining a possible winner, or that an ill-defined ballot $b^{-1}$ will be used in line 13 of the protocol at time $t=0$. However, Theorem~\ref{thm:cuds_stop}, presented in Section~\ref{sec:cud_theory}, shows that a CUD always terminates with a single alternative in $C$ declared a winner. We distinguish between two ways a CUD can terminate. We say that a CUD has {\bf converged}, if it terminates with a $w\in C^+$ declared as the winner; otherwise, the default alternative is chosen (the default is the worst option from the point of view of all voters).

\begin{algorithm}
\floatname{algorithm}{Game Protocol}
  \caption{Consensus Under Deadline\label{cud_game}, $CUD(\calF, \tau)$}
\begin{algorithmic}[1]
\Statex
\Statex \hspace{-6pt}\begin{tabular}{ll}
{\bf Input:} &MDVR $\calF$, deadline timeout $\tau$\\
&Set of voters $V$, set of alternatives $C$,\\
&Truthful profile $\veca$, and utilities $u_i$\\
{\bf Output:} &winner candidate (possibly default) $w\in C$\\
{\bf Initialise:} &Set $t\leftarrow\tau$, and $b_i^\tau\leftarrow\topc_i(C)$ for all $i\in V$\\\end{tabular}
\While {$t\geq 0$}
\State Ballots $b_i^t$ are declared\label{alg:line:ballot_declare}
\State All voters calculate $\vecs^t=\score(\vecb^t)$\label{alg:line:score_calc}
\If {$\calF(\vecs^t,t)=\{w\}$ for some $w\in C$}
\State \Return $w$ as the winner \label{alg:line:winner_stop}
\Comment Game stops
\EndIf
\For {$i\in V$}
\State $w_i\leftarrow\arg\max\limits_{c\in C}u_i(\vecs^t-b_i^t+c,t-1)$\label{alg:line:best_response} \Comment{Ties are broken by $a_i$}
\EndFor
\State $I\leftarrow\{i\in V| w_i\neq b_i^t\}$
\State $j\leftarrow Random(I)$\label{alg:line:random}\Comment {Random voter choice}
\For {$i\in V$}
\State $b_i^{t-1}\leftarrow b_i^t$
\EndFor
\State $b_j^{t-1}=w_j$
\Comment{Only $j$ is allowed to re-vote}
\State $t\leftarrow t-1$
\EndWhile
\end{algorithmic}
\end{algorithm}

\subsection{Quality of CUD outcomes}

In order to analyse the quality of the result of a CUD game, features of voting processes can be  adapted. One such feature is the \emph{additive Price of Anarchy ($PoA^+$)} \citep{branzei2013bad}. We adapt the $PoA^+$ to CUDs, requiring that at least some valid alternative can become a winner, as follows.
\begin{definition}\label{def:add_poa}
Let $\veca$ be the truthful profile of voters participating in a CUD, $\vecb$ a ballot profile induced by $\veca$ (i.e., $b_i=\topc_i(C)$), and $\vecs=\score(\vecb)$. Denote all valid candidates that the CUD may converge to by $\widehat{C}\subseteq C^+$. 
Then the CUD's \emph{additive Price of Anarchy} is\\
\centerline{$PoA^+(\veca) = \max\limits_{c\in C}s_c-\min\limits_{c\in\widehat{C}}s_c$}
\end{definition}

Namely, $PoA^+$ is the score of the least-preferred valid alternative that could become the winner of a CUD, subtracted from the score of the truthful winner. Notice that $PoA^+$ is well defined only for games where $\widehat{C}$ is not empty, i.e., there is at least one non-trivial conclusion to the voting process.\footnote{Here we follow well-established definitions of the additive Price of Anarchy for voting processes with restricted dynamics (see,  e.g.,~\cite{branzei2013bad,omprj_2015_aaai}).} 


\section{Theoretical Features of CUD}\label{sec:cud_theory}

We begin our theoretical analysis by showing that the CUD Protocol is stable, i.e., if the set of possible winners is not empty at the beginning then one single winner is found and the process cannot end in a default. In all our proofs time appears to be running backwards, as we measure the number of steps until the deadline. Thus, with every decision step taken by voters, time winds \emph{down} from $t=\tau$ to $t=0$.
    

\begin{theorem}\label{thm:cuds_stop}
In any $CUD(\calF^{IMaj}_\sigma, \tau)$ with $\sigma\in(\frac{n}{2},n]$ and consistent utility functions, either the set of possible winners at time $t=\tau$ is empty and the outcome will be set to the default alternative, or a valid (i.e., non-default) alternative becomes the winner at time $t\in[0:\tau]$.
\end{theorem}

Notice that Theorem~\ref{thm:cuds_stop} applies to both voter types of interest, lazy and proactive. In fact, all of our theoretical results are applicable to both of these types.
Now, having established that the algorithm always stops, we can study the  CUD process in detail. All proofs can be found in the appendix. 

First, we show that if a candidate is not a possible winner at stage $t$ (e.g., now), then he will never become a possible winner at any further stage $t+i$. 

\begin{lemma}\label{lemma:not_in_PW_before}
Let $c \notin \widehat{W}$ at step $t$, then $c \notin \widehat{W}$ at any step $t'<t$.
\end{lemma}

The next lemma suggests that a candidate that receives an additional vote from one time step to the next must be in the set of possible winners.
\begin{lemma}\label{lemma:if_1_vote_more}
If $s_c^{t+1}<s_c^t$, then $c \in \widehat{W}(s^t,t)$.
\end{lemma}

The next two lemmas prove that a voter 
always votes for her top choice in the set of possible winners; and, if a voter switches her support to another candidate, then this newly supported candidate has at least the same score (number of votes) as the previously supported candidate. 

\begin{lemma}\label{lemma:vote_for_top}
If a voter $i$ at the time step $t$ votes for candidate $c \in \widehat{W}(s^t,t)$, then  $c=\topc_i(\widehat{W}(s^t,t))$.
\end{lemma}

\begin{lemma}\label{lemma:if_change_vote_score_changes}
If there is a voter that changes her ballot at time $t$ from voting for $c$ to voting for $c'$, then $s_c^{t+1} \leq s_{c'}^{t+1}$.
\end{lemma}

Now we can place a condition on the game features that ensures that the default alternative is never set. 

\begin{cor}\label{thm:finer_bound}
Let $\veca=(a_1,\ldots,a_n)$ be the truthful profile, let $\tau$ be the deadline time, and let $\vecb$ be the ballot profile induced by $\veca$, i.e., $b_i=\topc_i(C)$. CUD stops with some $w\in C^+$  if and only if there is an alternative $c\in C^+$ so that $\score_c(\vecb)\geq \sigma-\tau$.
\end{cor}

Corollary~\ref{thm:finer_bound} essentially provides a finer bound on what the initial scores must look like, so that CUD converges. Intuitively, it says that for a non-default alternative to become the declared winner, there must be enough time for it to gather additional support to achieve the majority threshold. However, Corollary~\ref{thm:finer_bound} does not guarantee that a \emph{particular} alternative will be declared as the winner. For such a guarantee, a much more stringent condition must be required of $\tau$ and $n$, as the following theorem states.

\begin{theorem}\label{thm:single_winner}
Let $\veca=(a_1,\ldots,a_n)$ be the truthful profile, let $\tau$ be the deadline time, and let $\vecb$ be the ballot profile induced by $\veca$, i.e., $b_i=\topc_i(C)$). If there is an alternative $c\in C^+$ so that $\score_c(\vecb)\geq\max\left\{\left\lfloor\frac{n}{2}\right\rfloor+1,\sigma-\tau\right\}$ then CUD terminates with $c$ as the winner.
\end{theorem}

The following example demonstrates that the bound of Theorem~\ref{thm:single_winner} is tight. That is, if the score were any lower than Theorem~\ref{thm:single_winner} suggests, then, for at least some truthful profiles, a CUD would have more than one alternative that could be declared the winner.

\begin{example}\label{ex:single_winner_tight}
Let the number of voters $n$ be even, and the number of (non-default) alternatives $m\geq 2$. Furthermore, we assume that $\tau>\frac{n}{2}$. Construct the truthful profile $\veca$ so that $\topc_i(C)=\begin{cases}c_1& i\in[1:\frac{n}{2}]\\ c_2&i\in[\frac{n}{2}+1:n]\end{cases}$, where $c_1,c_2\in C^+$. All other non-default candidates may appear in $a_i$ in any order. Then, both $c_1$ and $c_2$ can possibly be declared as the winner in a CUD.
\end{example}

Having dealt with the characterisation of valid alternatives that are either guaranteed to, or can potentially, be declared a winner, we can exploit this knowledge to place some bounds on the additive Price of Anarchy for CUDs.

\begin{theorem}\label{thm:poa_bounds}
Let $\veca$ be the truthful profile of voters participating in a CUD. Assuming that it is well-defined for the CUD instance, the following bounds can be placed on the additive Price of Anarchy, $PoA^+$, depending on the ratio of the deadline timeout $\tau$ and the number of voters $n$:

\begin{enumerate}
\item If $\tau \leq \sigma - \left\lfloor\frac{n}{2}\right\rfloor$, then \begin{equation}PoA^+(\veca)=0.\label{eq:poa_bound_1}\end{equation}
\item If $\sigma-\left\lfloor\frac{n}{2}\right\rfloor<\tau<\sigma$, then \begin{equation}PoA^+(\veca) \leq \left\lfloor\frac{n}{2}\right\rfloor+\tau -\sigma.\label{eq:poa_bound_2}\end{equation}
\item If $\tau\geq \sigma$, then
\begin{equation}PoA^+(\veca) \leq \left\lfloor\frac{n}{2}\right\rfloor-1.\label{eq:poa_bound_3}\end{equation}
\end{enumerate}
\end{theorem}

\begin{lemma}\label{lemma:poa_bounds_tight}
The last two bounds in Theorem~\ref{thm:poa_bounds} are tight. For all $\tau$ and $n$ that satisfy the conditions of Equations~\ref{eq:poa_bound_2} and~\ref{eq:poa_bound_3}, there exists a truthful profile $\veca$ such that the corresponding bound holds as an equality.
\end{lemma}

\section{Experimental Features of CUD}\label{sec:simulations}
The experimental setting varied according to the following parameters: 
    \begin{itemize}
        \item \textbf{Voter types: } Lazy or proactive voters. 
        \item \textbf{Data sets: } A total of eight data sets: four simulated data sets and four real-world ones. The simulated data sets contain all permutations over 5-8 candidates (i.e., impartial culture for 5-8 candidates). These data sets are titled ``Uniform5", ``Uniform6", ``Uniform7" and ``Uniform8". The four real-world data sets are: the Sushi data set (5000 voters, 10 candidates) \citep{kamishima_etal_2005}, the T-shirt data set (30 voters, 11 candidates), the Courses 2003 data set (146 voters, 8 candidates) and the Courses 2004 data set (153 voters, 7 candidates). The three latter data sets are taken from the Preflib library~\citep{mattei_walsh_2013}. 
        \item \textbf{Number of voters:} We examined a fixed number of $n$ ($n=10$, $20$ or $30$) voters.
        \item \textbf{Voter preferences sets:} We created 20 random sets of voter preferences by sampling with return from each data set. For example, for $n=10$ the same experiment was conducted 20 times, each time with another set of preferences for each of the voters. 
        \item \textbf{Time until the deadline: } For 10 voters: $\tau = {2,...,11}$. For 20 voters: $\tau = {2,...,21}$. For 30 voters: $\tau = {2,...,31}$.
    \end{itemize}
    
    The various combinations of the parameters result in 9600 different \textit{experimental settings}. Each experimental setting contains some randomness: it may be that several voters wish to update their vote in a given round. In such a case, according to Protocol 1, only one of these voters is chosen at random. In order to verify that this randomness does not effect the stability of the experimental setting, we generated 30,000 election runs \textit{for each} experimental setting. 
    We examined:
    \begin{itemize}
    \item \textbf{Convergence:} The minimum deadline time $t$ which is necessary for \textit{all} experiment runs to converge in a given experimental setting.  
    \item  \textbf{Required number of vote changes:} How many vote changes are required for the process to converge. Here, we provide averages and standard deviations over the 20 random voter preferences sets (which are already averages of the 30,000 runs). \item \textbf{The Additive Price of Anarchy ($PoA^+$):} 
    Recall that $PoA^+$ is defined as the score of the least-preferred valid alternative that could become the winner of a CUD, subtracted from the score of the truthful winner (Definition~\ref{def:add_poa}). In order to receive an approximation of $PoA^+$ for each experimental setting, we looked at the winners in the 30,000 experiment runs. Out of these winners, we found the winner that is the least preferred alternative. The least preferred alternative found in these 30,000 runs was subtracted from the truthful winner. To avoid overloading notation, we refer to this $PoA^+$ approximation simply as $PoA^+$. Note that $PoA^+$ can be computed only for processes where some winner is found, i.e., processes that converged. For processes where the winner is not found, the reader should refer to the convergence rate.  
    \end{itemize}

Our theoretical results are relevant for all majority thresholds $\sigma\in(\frac{n}{2},n]$. However, to investigate the finer features of CUDs, we concentrate on Unanimity, fixing $\sigma$ to the number of voters $n$. Even though it means that we experimentally study an extreme CUD case, fixing $\sigma$ allows us to exclude it as a free parameter, and concentrate on studying more complex game features, such as the Additive Price of Anarchy.

In order to conclude which voter type performs best over multiple data sets, we followed a robust non-parametric procedure proposed by Garc\'ia et al. \citeyearpar{Garcia10}. This procedure allows us to avoid the assumption that the performance difference between voter types is normally distributed, making it more adequate than the standard t-test. We first used the Friedman Aligned Ranks test in order to reject the null hypothesis that all voter types perform the same. This was followed by the Bonferroni-Dunn test to find whether one of the types performs significantly better than the other.

\subsection{Simulation Results}\label{subsec:sim_results}

{\bf Convergence:} As Theorem~\ref{thm:finer_bound} indicates, convergence always occurs when there is enough time until the deadline to allow voters to change their vote (i.e., when the number of iterations is larger than the number of voters). Interestingly, we find that in all experimental settings the process converges faster than the worst-case convergence time, described by Theorem~\ref{thm:finer_bound}. The convergence of Uniform data sets is slower than that of real-world data sets, but still faster than that of the theoretical bound. 
    This suggest two things:
    (a) Real-world voter preferences are a-priori aligned, in a sense, and it is easier to reach consensus than with an arbitrary preference profile.
    (b) Worst-case profiles are extremely rare even in impartial cultures. 
Table~\ref{table1} shows the minimum deadline time $t$ which is necessary for \textit{all} experiments to converge.
In our experiments, this time was confirmed to be identical for proactive and lazy voters.
In the real-world data sets, the process seems to converge faster when there are fewer candidates. The exact impact of the candidate number on convergence should be examined on a wider variety of data sets; we leave this for future research.
The more voters, the longer it takes the process to converge. This is illustrated in Figure~\ref{convergence1}, on the Courses 2004 data set and T-shirt data set 
(results are similar for other data sets). The figures differ from one another because the number of candidates in each of the data sets is different.

\begin{table}
\caption{Process convergence times: The minimum deadline time $t$ which is necessary for all experiments to converge. The times are identical for both proactive and lazy voters. Each column represents a different number of voters, and each row a different data set. The number of candidates in the data set is indicated in brackets. For example, for the Courses 2004 data set (row 2), for 10 voters (column 1), all of the experiments converged when the deadline $\tau \geq 6$.}
\label{table1}
\begin{small}
  \begin{center}
    \begin{tabular}{|c|c|c|c|}\hline
      \textbf{Data set $\Downarrow\backslash\Rightarrow$ Number of Voters}&\textbf{10}&\textbf{20}&\textbf{30}\\\hline
      \textbf{Uniform5 (5)}&7&16&23\\
      \textbf{Uniform6 (6)}&8&16&24\\
      \textbf{Uniform7 (7)}&8&16&24\\
      \textbf{Uniform8 (8)}&8&16&25\\
      \textbf{Courses 2004 (7)}&6&13&19\\
      \textbf{Courses 2003 (8)}&7&14&23\\
      \textbf{Sushi (10)}&8&14&24\\
      \textbf{T-Shirts (11)}&8&14&23\\\hline
    \end{tabular}
  \end{center}
\end{small}
\end{table}

{\bf Required number of vote changes:}
Table~\ref{table2} shows, for different data sets and varying number of voters, the normalized average of vote changes required to reach a consensus. The presented results include one simulated data set (Uniform5) and the four real-world data sets. Full results are presented in Appendix \ref{appendix:uniform}. Proactive voters require more vote changes, this is especially true in the real-world data sets.

\begin{table} 
\caption{Number of vote changes: the average  and standard deviation of the number of vote changes required to reach a consensus.}
\label{table2}
\begin{small}
\begin{center}
   \begin{tabular}{| c| c c| c c| }
    \hline
    & \multicolumn{2}{c |}{\textbf{10 voters}} & \multicolumn{2}{c |}{\textbf{20 voters}} \\
     \textbf{Data set} & \textbf{lazy} & \textbf{proactive} & \textbf{lazy} & \textbf{proactive}  \\
     \hline
    \textbf{Uniform5} & $5.06  \pm  2.95$ & $5.15  \pm 3.02  $ & $10.74  \pm 6.02  $ & $10.83 \pm 6.1   $   \\
    \textbf{Courses 2003} & $4.98 \pm 2.45 $   & $6.46  \pm 3.54 $  &$ 11.31 \pm 5.45  $ & $13.09  \pm 6.72  $  \\
    \textbf{Courses 2004} & $3.55 \pm 2.05  $ & $5.73  \pm 3.78 $ & $7.5  \pm 4.15  $&$ 11.19  \pm 7.04 $  \\
    \textbf{Sushi} & $5.71 \pm 2.9 $ & $ 6.75 \pm 3.7 $ &$ 11.11  \pm 6.32 $ & $12.42 \pm 7.31 $ \\
    \textbf{T-shirts} & $ 5.66 \pm 2.5  $ & $ 6.91 \pm 3.45  $ & $12.38 \pm 5.97  $ & $13.76 \pm 6.93   $   \\
    \hline
    \end{tabular}
    
    \begin{tabular}{| c| c c|   }
    \hline
    & \multicolumn{2}{c |}{\textbf{30 voters}}   \\
     \textbf{Data set} & \textbf{lazy} & \textbf{proactive}   \\
     \hline
    \textbf{Uniform5}   & $19.46 \pm 7.6     $& $19.49 \pm $ 7.63 \\
    \textbf{Courses 2003}   &$ 18.7 \pm 6.7 $    &$ 21.17 \pm 8.22$\\
    \textbf{Courses 2004} &  $11.98 \pm  6.37 $ & $ 17.13 \pm 10.25$\\
    \textbf{Sushi}  & $18.61 \pm 8.14 $ & $20.71  \pm 9.52 $\\
    \textbf{T-shirts}   & $18.74 \pm 9.43    $ & $ 20.43 \pm 10.6 $ \\
    \hline
    \end{tabular}
\end{center}
\end{small}
\end{table}

According to the Friedman test, there is a significant difference between the number of vote changes performed by purely lazy and purely proactive voter sets.
According to the Bonferroni-Dunn test, proactive voters require a significantly higher number of vote changes in order to reach consensus.
Notice, however, that the fact that proactive voters need more vote changes to converge does not affect the convergence time. This aligns well with our theoretical findings. In particular, Theorem~\ref{thm:single_winner} and Corollary~\ref{thm:finer_bound} relate protocol convergence and the deadline independently of voter types. 

{\bf Additive price of anarchy:} The upper bound to the Additive Price of Anarchy was computed as the plurality score of the least-preferred candidate that was elected to be a unanimous winner in one of the 30,000 experiments, subtracted from the plurality score of the truthful winner. Note that in order to determine the real price of anarchy we must pursue every possibility. As 30,000 experiments were used, our result is actually an upper bound to the price of anarchy. 
For example, consider 12 voters and 4 candidates and the following scores at the beginning of the process: 
$s_{c_1}^{\tau}=2, s_{c_2}^{\tau}=6, s_{c_3}^{\tau}=1, s_{c_4}^{\tau}=3$. The truthful plurality winner is $c_2$ with a score of 6. If in one of the experiments $c_3$ is the unanimous winner, the additive price of anarchy is 5. If $c_3$ does not win in any of the experiments, but in some of them $c_1$ wins, the price is 4. 
Table~\ref{table3} shows the average and standard deviations of the upper bounds for the additive price of anarchy for 10 voters for different initial $\tau$ (time until the deadline). These are averages (and standard deviations) over 20 different voter preference sets. 
Results for 20 and 30 voters as well as for all simulated data sets (Uniform6, Uniform 7 and Uniform 8) are provided in Appendix \ref{appendix:uniform}.
The $PoA^+$ is identical for both lazy and proactive voters. 
As can be inferred from the table, the $PoA^+$ is low, meaning that in most cases, the plurality winner was the winner. 
\begin{table}
\caption{The average and standard deviations of the upper bounds for the additive price of anarchy for 10 voters for different initial $\tau$ (time until the deadline). The $PoA^+$ is identical for both lazy and proactive voters. }
\label{table3}
\begin{small}
\begin{center}
  \begin{tabular}{|L|c|c|c|c|c|}\hline
      \textbf{Time until the deadline}&\textbf{Uniform5}&\textbf{Courses 2003}&\textbf{Courses 2004}&\textbf{Sushi} &\textbf{T-Shirts}\\\hline
    \textbf{2}   &$0 \pm 0$  & $0 \pm 0$  &$0 \pm 0$  &$0 \pm 0$  &$0 \pm 0$  \\
    \textbf{3}    &$0 \pm 0$  & $0 \pm 0$  &$0 \pm 0$  &$0 \pm 0$  &$0 \pm 0$   \\
    \textbf{4}    & $0 \pm 0$  &$0 \pm 0$  &$0 \pm 0$  &$0 \pm 0$  &$0 \pm 0$    \\
    \textbf{5}    & $0 \pm 0$  &$0 \pm 0$  &$0 \pm 0$  &$0 \pm 0$  &$0 \pm 0$    \\
    \textbf{6}   &$0 \pm 0$  & $0 \pm 0$  &$0 \pm 0$  &$0 \pm 0$  &$0 \pm 0$  \\
    \textbf{7} &$0.1 \pm 0.3$     & $0.15 \pm 0.36$  &$0.05 \pm 0.22$  &$0 \pm 0$  &$0.15 \pm 0.36$  \\
    \textbf{8} &$0.85 \pm 0.74$   & $0.4 \pm 0.59$  &$0.05 \pm 0.22$  &$0.15 \pm 0.48$  &$0.55 \pm 0.81$  \\
    \textbf{9} &$1 \pm 0.91$     & $0.7 \pm 1.02$  &$0.05 \pm 0.22$  &$0.45 \pm 1.04$  &$0.65 \pm 0.86$  \\
    \textbf{10}  &$1 \pm 0.91$   & $0.7 \pm 1.02$  &$0.05 \pm 0.22$  &$0.45 \pm 1.04$  &$0.65 \pm 0.86$  \\
    \textbf{11}  &$1 \pm 0.91$  & $0.7 \pm 1.02$  &$0.05 \pm 0.22$  &$0.45 \pm 1.04$  &$0.65 \pm 0.86$  \\
    \hline
    \end{tabular}%
\end{center}
\end{small}
\end{table}

\begin{figure}[ht]
  {\center
    \includegraphics[scale=.27]{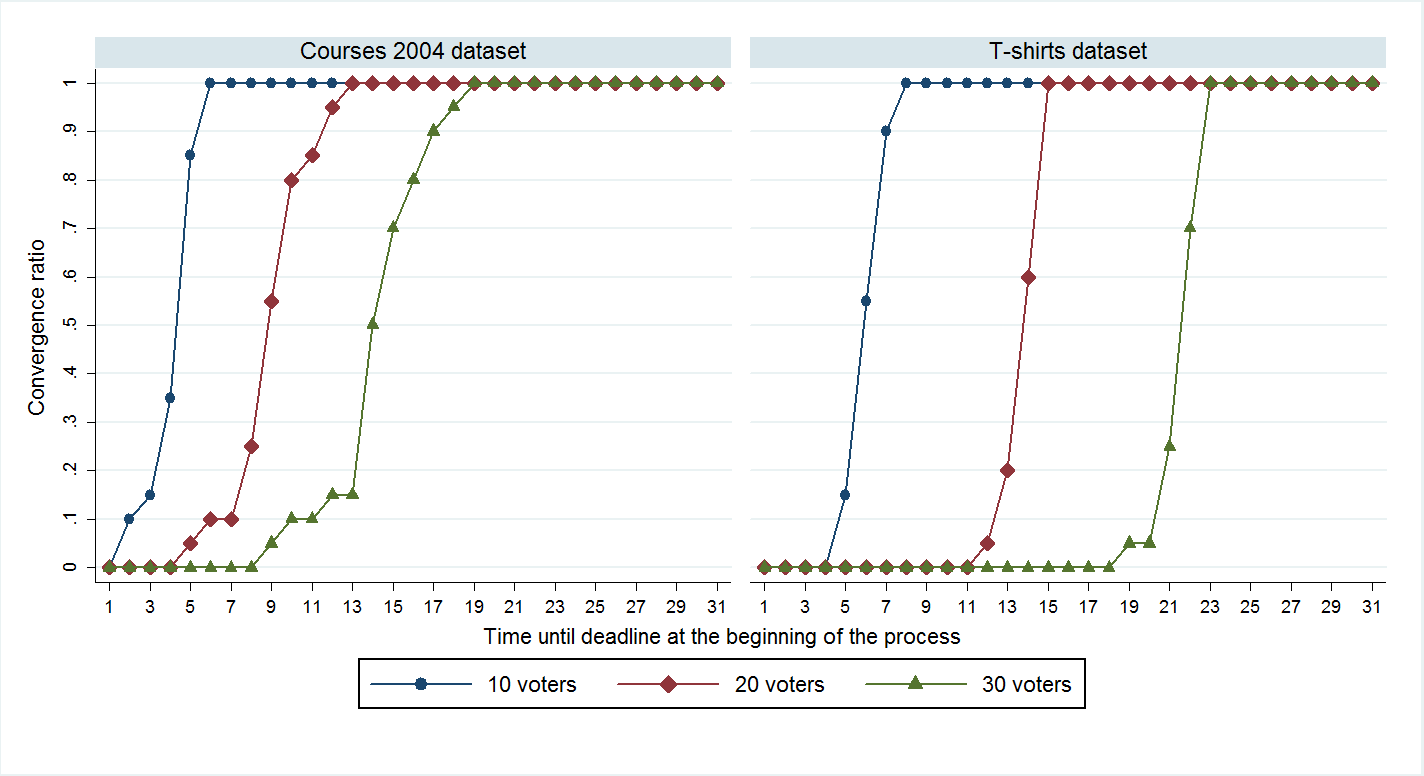}
    \caption{Voting process convergence ratio---Courses-2004 (left) and T-shirt (right)  data sets. On the x-axis is the time $\tau$ that the process begins; on the y-axis, the ratio of games that converged in that time, out of all experiments in the same sub-class (the sub-class of experiments are all experiments from a data set with the same number of voters and the same $\tau$).}
    \label{convergence1}\label{convergence2}
  }
\end{figure}


To conclude, it is interesting to note that although the convergence ratio is equal for both voter types, the required number of vote changes is \emph{higher} for the proactive, rather than lazy, voters. Furthermore, the additive price of anarchy is equal for proactive and lazy voters. We thus did not see any reason to continue with the proactive voters, who seem to be inferior to the lazy voters. 
Therefore, when designing the bots in user studies (in the next section and in  \cite{,Gvirts_Dery}), only the lazy voter architecture is used.

\section{CUD User study}\label{sec:userstudy}
The purpose of the user study was to examine whether human voters play in a rational manner (more on our specific definition of the term ``rational'' below). We built a game, called the CUD game which follows our CUD model structure. The game was played by students as well as rational bots that were programmed according to the lazy bot behavior described in Section \ref{sec:model}. We herein describe the game and the data collection method, and analyze the results. 

\subsection{CUD-Game flow} \label{sub:game_flow} 
The framework contains 2 modules:
\begin{enumerate}
\item The \emph{CUD-Game}\footnote{https://github.com/DavidBenYosef/CUDGame}: a web-based interactive multi-player decision game designed to facilitate an iterative group decision process.
The game was implemented as a Java server-client system with an online multilingual HTML with Javascript interface. All voter's actions were collected and saved to a MySQL database.
The game used a full-duplex asynchronous communication (Websocket protocol), to enable a fully interactive game between all voters.
\item The \emph{CUD-Runner}\footnote{https://github.com/DavidBenYosef/CUDRunner}%
: a standalone Java application that received the data collected by \emph{CUD-Game} runs, and allowed us to simulate and analyze the games.

\end{enumerate}

\begin{figure}[ht]
\centering
\includegraphics[scale=0.4]{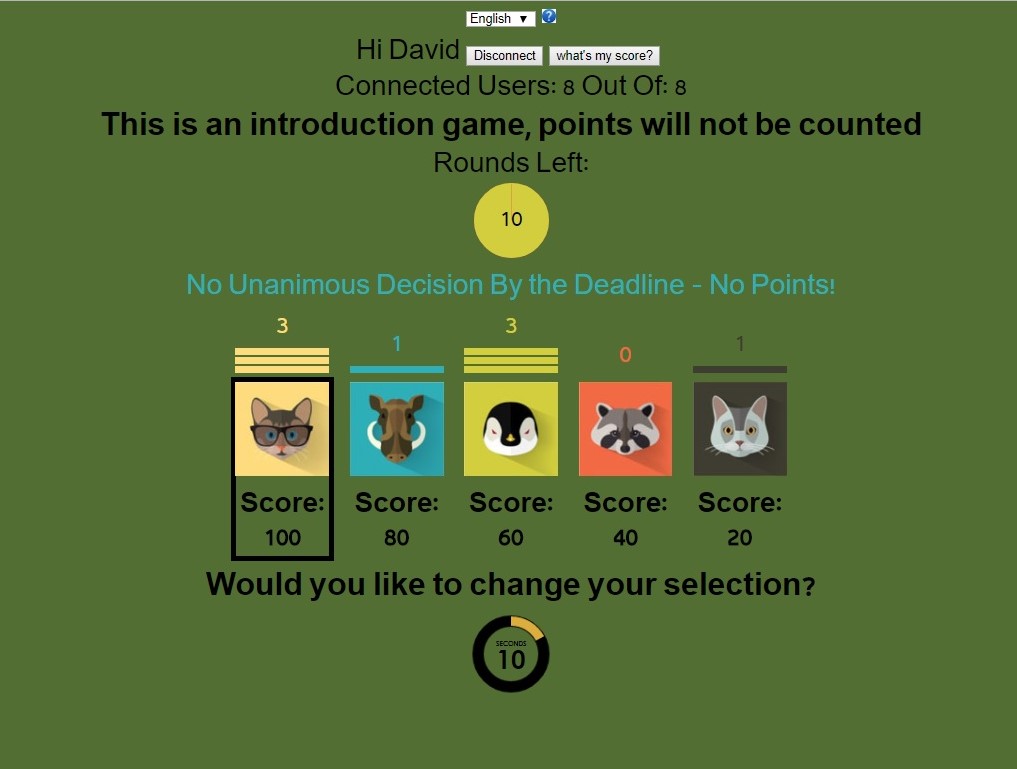}
\caption{The screen at the beginning of the game. The voter's highest preference (the cat with the glasses) is selected automatically.  }
\label{fig:game}
\end{figure}

\begin{figure}[ht]
\centering%
\includegraphics[scale=0.4]{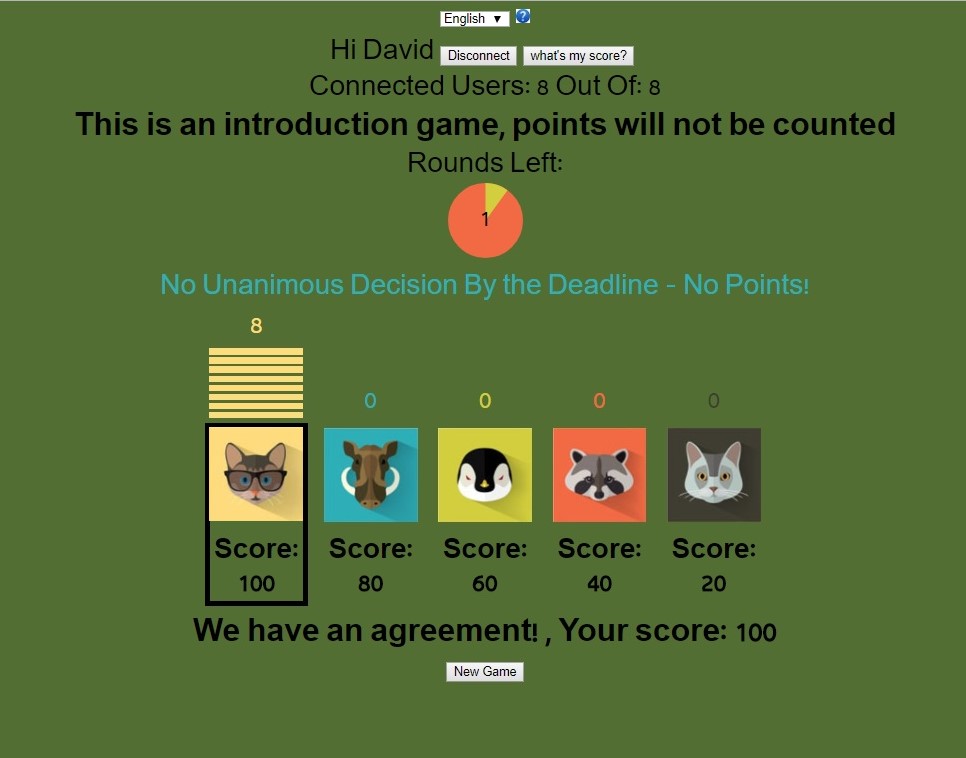}
\caption{An example of a screen at a possible end of the game. In this example, a consensus was reached; all voters agreed on the cat with glasses. }
\label{fig:finish}
\end{figure}

The \emph{CUD-Game} proceeds as follows:
first, the game is activated at an agreed time, so that all voters will play at the same time (while the game is off, the voters are not able to log-in, and they receive an appropriate message).
Once the game is on, the users are required to confirm their participation in the research. 
Next, each voter logs-in 
with a name and an identification number, and waits for the other voters to join the game. 
The game begins once the predefined number of voters is reached. When the number of logged voters exceeds the number of voters for a game, other games begin simultaneously and independently.
At the beginning of the game, each voter receives a predefined preference profile chosen uniformly at random. The highest preference for each voter is selected automatically and the current result is shown (Figure \ref{fig:game}).

On each round, each voter decides whether to change her current selection. Voters that want to change their selection do so by selecting another alternative. The voters are required to reply within a fixed time-span (usually set to $15$ seconds).
Next, the system randomly selects one of the voters who applied for a ballot change and computes the new intermediate results. The round ends. The system checks whether the deadline has been reached and whether a consensus has been reached.  When the answer to both of these questions is negative, a new round begins with the display of the current ballots. The participants are aware at all times of the number of remaining rounds.
If no consensus is reached by the deadline, all voters receive 0 points. When a consensus is reached, each voter receives a score corresponding to the chosen preference (Figure \ref{fig:finish}). 

\subsection{Reward method}
\label{sub:reward}
One of the fundamental requirements from the \emph{CUD-Game} is that no consensus is the worst option for all voters, and that there will be no reward if consensus is not reached.
This is the reason why we decided to avoid using the Amazon Turk platform as done by e.g.~\cite{meir2020strategic,mao2013better,scheuerman2019heuristic,mennle2015power} and many others, since this platform requires paying the players for logging in, regardless of their performance. In particular, in spite of the common measures~\citep{mason2012conducting},  ~\cite{meir2020strategic} faced ``ghost'' players that played only for the ``show-up payment'' (that was given to lab participants and players on the Amazon Turk platform) and thus introduced much noise into the collected data.
In our research, we defined the term of \emph{game points} to represent an abstract reward layer.
In each game, the voter had the option to gain between 0 to 100 game points, e.g., when there are 5 candidates, the score if the first preference is selected is 100 game points, the score for the second preference is 80 game points, and so on, so that the score of the least-preferred candidate is 20 game points. If there is no consensus by the deadline, the voter receives $0$ game points. For example, if a certain candidate is chosen, and a voter has this candidate ranked as her second preference, then she will receive 80 points for this game.

The score for each voter is collected and saved, and can be converted and paid in any relevant method later on. 
In our user-study we had undergraduate students play the game, and the scores were converted to course bonus points (200 game points were equivalent to 1 bonus point). Students could earn up to a total of five bonus points, that were added to their course grade. 
Thus, a student that received a grade of 95 in the course, and earned four bonus points, received 99 as a final grade. A student that earned five bonus points and received 88 in the course, received 93 as the final grade. 
Students had a high incentive to participate in the game, but participation was not obligatory; a student could receive a perfect score (100) in the course simply by completing all obligatory course work and receiving a perfect score (100) in the final exam. 

Our study received ethical approval of the institutional review board of Ariel University, Israel. The review board explicitly investigated the question of enforcement vs. incentivization. They agreed that our scheme is  of positive incentivization rather than negative enforcement. 
 
\subsection{Irrational voters}\label{sec:irrational}
In our context, rational behavior is the tendency to conform to the group by voting for a candidate that was chosen by the majority of the group.  In contrast to this rational behavior, irrational behaviors are actions that do not promote consensus-building. We followed the two types of irrational voting defined by  \cite{Gvirts_Dery}: 

\begin{definition}
{\em{Opposing alignment (OA)}:} a voter who opts to change her vote to a candidate $c_j \in C$ that she ranks lower than $c_{maj}$, the candidate with the current highest score. Namely, the voter prefers $c_{maj}$ over $c_j$, but still would like to vote for $c_j$.
\end{definition}

Consider the player in figure \ref{fig:player_3}. An OA action would be to vote for one of the cards that has a lower value than the penguin (values 40 or 20). 
If the voter realized that her current selection will not win, and decided to change her selection, she should at least change to $c_{maj}$ and not to a less preferred candidate.

\begin{definition}
{\em{Inappropriate alignment (IA):}} a voter who opts to change her vote to a candidate $c_j \in C$, when there are other candidates that the voter prefers over $c_j$, and furthermore have a score that is higher than the score of $c_j$. Namely, the voter opts to change her vote for a less preferred candidate with fewer votes.
\end{definition}
Returning to the player in figure \ref{fig:player_3}, the candidate preferred by most players is the penguin. This card has a value of 60 for player 3. An IA action would be to vote for the racoon which has a higher value (80) for him/her. 
This behavior is irrational, since there is no reason to change to a less preferred candidate, if this candidate also has a lower score.

\begin{figure}[ht]
\centering%
\includegraphics[scale=0.6]{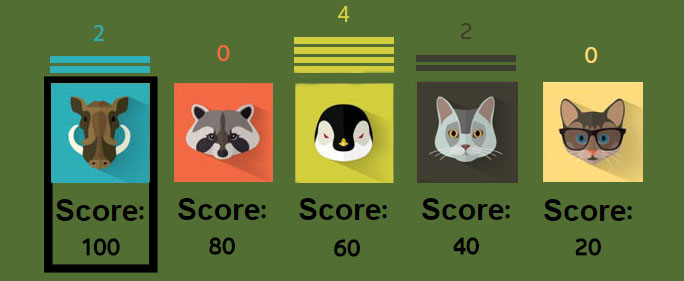}
\caption{The voter's options. The voter has previously voted for the blue card with the boar on it. The voter can choose if to stay with this choice, or vote for one of the other cards. The yellow card with the penguin currntly has the most votes - four votes.}
\label{fig:player_3}
\end{figure}


\subsection{Data Collection and Metrics}\label{sec:data}
Games consisted of eight voters each. 
We set a limit for the number of games per voter to 15 games, so that the voters would remain motivated to play and ``fight'' for every point. In an initial experiment (omitted from this paper for lack of interest), we found that when the number of games in unlimited, voters have no incentive to play. Instead, in games in which a player perceived the setting as unfavorable, he/she preferred to wait, knowing that at some stage other games with better setting will be played. 

A total of 72 students played a total of 264 games: 144 mixed games with two bots and six students, of which 137 converged, and 120 games with no bots, of which 105 converged. Note that the students had no idea whether they were playing against other students or against bots.

We also ran a set of bot-only games, eight bots per game. A total of 10000 bots played 1250 games. As expected, all of the games converged.

We found no significant difference in the number of irrational actions performed in the first two games of each student vs.~later games. We thus concluded that the irrational actions are not due to misunderstandings.
We compared between games with irrational actions ($OA$ or $IA$) and games without irrational actions (rational games). 
We defined irrational games as games where at least one irrational action was performed, and rational games as games where no irrational actions were performed.
We examined all the data collected in the case studies: students only games, mixed games of students and bots, and bot-only games.

We examined: 
\begin{itemize}
\item \textbf{Convergence percentage}
\item\textbf{ Average reward points} - how do the users' actions affect the average reward points received in a game?
\item \textbf{The price of Reality ($PoR$)} - the CUD-Game is played by real voters that are not necessarily rational. The voters might decide to agree on choosing the worst option, and thus each game has a potential to reach the highest possible $PoA^+$. Furthermore, each game is played only once (and defiantly not $30,000$ times as in the simulated experiments). Thus, we cannot measure $PoA^+$, and we  dubb what is actually measured as ``the Price of Reality" -- $PoR$. $PoR$ is defined as the score of the outcome of the CUD-Game subtracted from the score of the truthful winner. 
\item \textbf{Convergence time} - how do user (and bot) actions affect the time it takes the game to converge?
\end{itemize}

\subsection{Game Results}\label{sec:results}
In order to determine whether our results obtain a statistically significant difference with respect to the four different phases and with respect to the game state (rational or irrational), we performed a two-way ANOVA analysis. All results were found to be significantly different with $pvalue<0.05$.

Figure \ref{fig:convergence} shows the influence of irrational actions on the convergence of the games.
Axis y displays the percentage of converged games out of all games.
As expected, we see that irrational actions lower the chances to reach an agreement by $10\%-15\%$.
We also see that bot-only games always converge, and that mixed games converge in higher percentages than student-only games, since less voters can play irrational actions (98\% for mixed games with rational actions only and 88\% for mixed games with irrational actions, 89\% for student games with rational actions only and 82\% for student games with irrational actions).

\begin{figure}
\centering%
\includegraphics[scale=0.5]{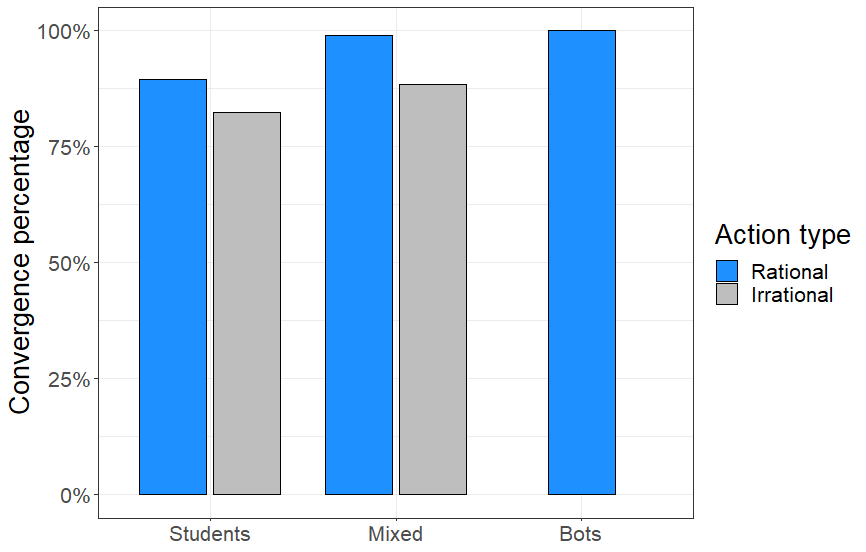}
\caption{The influence of irrational actions on the convergence of the games according to game type (Students only, Mixed students and bots or bots only)  for games with only rational actions and games that include irrational actions. Axis y displays the percentage of converged games out of all games. }
\label{fig:convergence}
\end{figure}

\begin{figure}
\centering%
\includegraphics[scale=0.5]{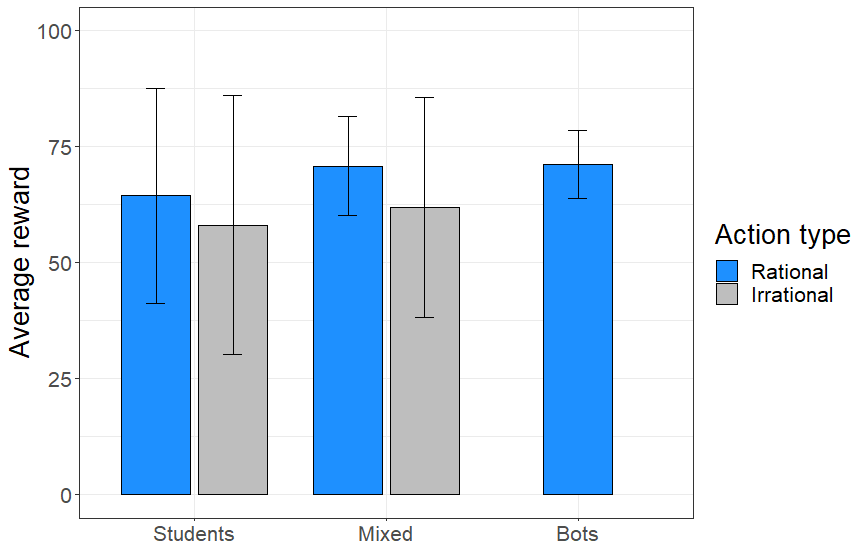}
\caption{Average reward points comparison according to game type (Students only, Mixed students and bots or bots only)  for games with only rational actions and games that include irrational actions. Bots are better than student players, even if they are rational student players. }
\label{fig:reward} 
\end{figure}

Next, we examined whether a high convergence percentage necessarily results in a higher satisfaction, as theoretically a few converged game with high average reward might produce higher satisfaction than many converged games with low average reward.

Figure \ref{fig:reward} shows the global satisfaction factor, i.e., the average reward points per game.
We notice a full correlation with the convergence percentage (Figure \ref{fig:convergence}); we see a consistent decrease in the satisfaction factor in games with irrational votes.
This is to be expected since when a voter plays irrationally, she plays against her interests, and more games do not converge, resulting in a decrease in satisfaction.

Our most important finding is found here: we see that the global satisfaction increases in mixed games (bots and student), and reaches its highest value in bot-only games. In other words, bots are better than student players, even if they are rational student players. Bots in a game (a student-bot game or a bot-only game) increase convergence percentage and increase the satisfaction of all voters, bots and students alike.

Note that bots outperform student players in rational games on both metric: bots converge faster and receive higher rewards. Although students in rational games do not perform any irrational moves, they still might decline to change their vote when needed and in this way block convergence. We did not log such declinations as irrational. We leave it for future research to study when such actions occur. 

Next, we examined the $PoR$ and the convergence time factors between converged games with rational actions and converged games with irrational actions. We avoided comparing games that ended with a default alternative, as the factors for these games are not defined. Since we have one winner for each game, we did not compute an approximation to the additive price of anarchy (as in Section \ref{sec:simulations}). Rather, we computed the score of the plurality winner in the truthful profile minus the score of the final winner in the game.

\begin{figure}
\centering%
\includegraphics[scale=0.5]{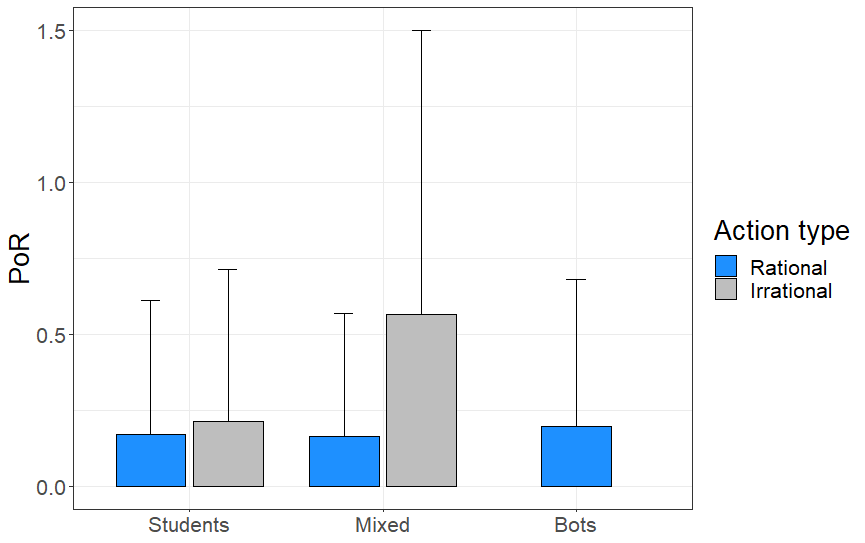}
\caption{Average $PoR$ according to game type (Students only, Mixed students and bots or bots only)  for games that converged with only rational actions and games that included irrational actions.}
\label{fig:poa}
\end{figure}

Figure \ref{fig:poa} presents the comparison of the average $PoR$ factor for rational and irrational games that converged.
The $PoR$ factor represents the distance between the truthful majority winner and the actual unanimity winner.
We see that irrational actions cause a higher $PoR$, since irrational voters play against their interests and may sometimes become the deciding factor, resulting in a consensus which is not the majority winner.
We also identify a big influence of irrational actions in the mixed games. After further investigation, we identified that when playing with irrational students, bots reached a higher average score.
We conclude that bots manage to ``recover'' from irrational actions better than students do, and to reach the best possible option considering the new situation.
Interestingly, we see that the $PoR$ of bot-only games is higher than the $PoR$ of games with rational students.  In fact, a high $PoR$ joint with a global satisfaction improvement is an indication of good strategic moves. Bots strategies change the majority winner to some other candidate, leading to an improvement in the global satisfaction.

\begin{figure}
\centering%
\includegraphics[scale=0.5]{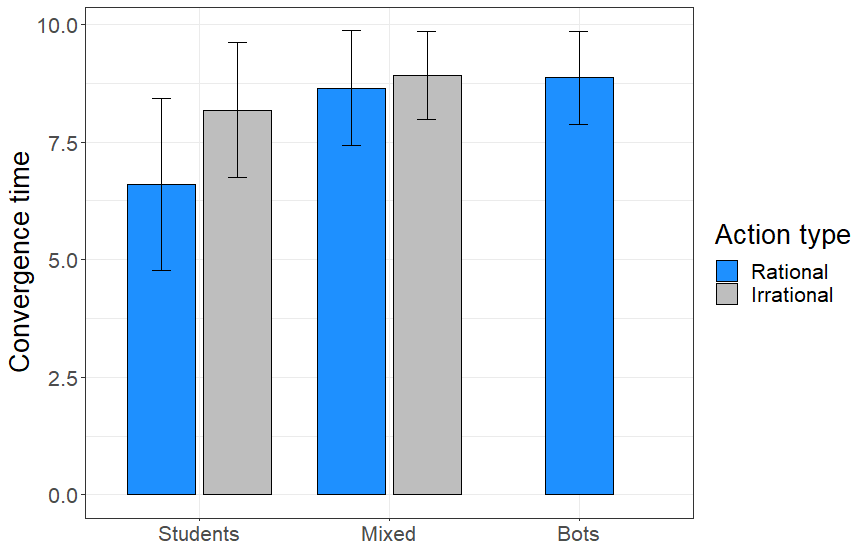}
\caption{The comparison of the convergence time factor according to game type (Students only, Mixed students and bots or bots only)  for games with only rational actions and games that include irrational actions.
The convergence time is the number of rounds the game took to converge. Convergence time 10 is the slowest, and it means the game converged only on the last round.}
\label{fig:time}
\end{figure}


Figure \ref{fig:time} demonstrates the comparison of the convergence time factor for rational and irrational games.
The convergence time is the number of rounds the game took to converge. Convergence time 10 is the slowest, and it means the game converged only on the last round.
We see that irrational actions cause the games to converge slower, as all other voters need to adapt their logic to the unexpected results.
Bot-only games converge the slowest, as bot logic makes them converge only when necessary.
As reported in figure \ref{fig:convergence}, bots always converge. Students do not always converge, but figure \ref{fig:time} shows that when they do converge - they do it faster than bots. 
Further research is required to determine why this is. Perhaps students converge faster since they are proactive, in contrary to the lazy bots. Maybe students are afraid of losing all the game points, so they prefer not to wait until the last round or may vote for a candidate with the largest number of votes (but with a smaller utility). 

{\bf Conclusion:}
Perhaps the bots' main weakness is that they operate under the assumption that all other voters (who are also bots) have the same logic. However, students do not hold the same assumption about other students. Games played by students converge faster since students are concerned that the game will not converge and end with the default alternative (the worst option for all involved), so they prefer not to wait until the last round. For the same reason, a single student playing against bots can manipulate them and change the outcome in her favor, meaning that the result is not necessarily the option preferred by the majority.
However, bots have a big advantage - they increase the convergence percentage and increase the satisfaction of all voters, bots and students alike.

\section{Discussion}\label{sec:conclude}
We present a novel model for an iterative voting process with a restricted number of iterations---\textit{Consensus Under a Deadline (CUD)}. 
We chose the sequential vote modification process because it is motivated by real-life scenarios (such as the school registration problem discussed in Section \ref{sec:intro}), and also since this is a widespread assumption in the iterative voting literature. 

After studying the  theoretical features of the CUD we find that when there are candidates whose lack of votes is smaller than the time until the deadline (there is enough time so that voters can switch their votes to this alternative), then the CUD converges with one of these candidates as a winner.  Moreover, if there is a candidate for whom  at least half of the voters voted, then the CUD converges with one of these candidates as a winner.   

Obviously not all individuals behave identically. We define two types of voters, proactive and lazy. Proactive voters are, in a sense, trigger-happy to change their vote, even if just to ensure that their preferred possible winner gets one more point. Lazy voters change their votes only when it is necessary to do so, i.e., when their vote is pivotal to keeping a particular alternative as a possible winner. 
Though it is natural to see proactive voters as those that actively seek consensus, they reap no benefit from their activism. Specifically, convergence time of both  proactive and  lazy voter CUDs are the same. On the other hand, as our experiments show, the number of vote changes until convergence is higher for proactive voters. In a way, they are inefficient in their behavior.


Theoretical results show $PoA^+$ principal bounds, but yield no specific trade-offs. Hence, our experiments also took a deeper look into additive Price of Anarchy as a measure of  winner quality. 
While re-confirming $PoA^+$ bounds experimentally, our experiments indicate that there is no difference between the $PoA^+$ of lazy and proactive voters. Thus, when designing voter bots, we designed them according to lazy voter behavior.

As a result of this study, bots that follow the theoretical model in Section \ref{sec:model} have been implemented in a \emph{CUD-Game}, a game developed in order to collect and analyze human behavior in consensus reaching scenarios with a tight deadline \citep{Yosef:2017}.
The CUD-Game was also used to study consensus reaching in students with Alexithymia \citep{Gvirts_Dery}. 

In the \emph{CUD-Game} we implemented voter bots, which follow the theoretical model in Section \ref{sec:model}, and compared it with data collected in a user study. We used the \emph{CUD-Game} to run a user study and examine how students reach an agreement, i.e., a consensus, within a tight deadline. The game enabled us to compare the students' behavior to rational (lazy) bot behavior. 

The results of decisions made only by bots, seem to have higher quality than when humans reach decisions. A reasonable conclusion is that when we want to know that a decision will converge, and that all voters will be as satisfied as possible with the result, then it is best for users to leave bots to reach a decision on their behalf. A user can notify her personal bot of her preferences, and let the bot represent her in the decision making process. Our results coincide with \cite{de2018social} who study interaction via agents, concluding that agents lead to fairer results.

Bots are playing ever-larger roles in our daily lives. We rely on them for navigation directions (e.g., Google Maps or Waze), we rely on them for information retrieval (e.g., any search engine), and we may one day rely on them when attempting to reach a group decision. This research is a step in that direction.


\textbf{Future Work:} One challenging extension is to adapt our model to use general Positional Scoring Rules (PSRs), e.g., veto, approval, and Borda rather than simple Majority. These rules would allow us to express more complex semantic structures over the set of alternatives, even if we stay with the committee example. For instance, committee members may need to decide on a candidate \emph{and} the pay grade. Some members may vote for a particular candidate, but veto a higher pay-grade, if they believe the candidate is not ready for a higher rank. At the same time, other committee members may be indifferent between two reasonable candidates.
In these cases, veto and approval voting are more appropriate than Majority.

Stepping even further away from our basic model, we need to investigate what information about the current vote affects a CUD's outcome, and in what way. For example, rather than using utilities that depend on an anonymous score vector, we can use weights to express the fact that the opinion of certain voters is more influential. The same technique would allow us to impose a price on the number of times a voter changes her ballot, e.g., her vote loses influence, being unstable. Simultaneous re-voting would be another research direction, entailing very non-trivial modifications.

We must also note a possibility to introduce a more global change in how we view a voter's strategic behaviour. Specifically, while we have justifiably, to an extent, assumed voter's ballot change to be myopic, higher cognitive hierarchies have drawn the recent attention of other researchers. We consider this to be an important issue to pursue among the extensions of our work, i.e., consider ballot selection that takes into account the contemporaneous reasoning of other voters. 


Finally, there are additional methods to measure the quality of an agreement point. While in this paper we chose to concentrate on Additive Price of Anarchy, various measures of social welfare are also popular in the voting literature, and we intend to pursue them as well. 

Of particular note is a possibility raised by one of our reviewers. Namely, that the vote aggregation function on which the deadline-limited protocol is based can be different from the aggregation that underpins the quality measurement of the agreement point\footnote{There are precious few works that already do so, such as~\cite{thompson2012empirical}, and the standard approach is to keep both protocol and evaluation aggregations the same.}. This can arise, for instance, as a result of an external practical limitation during system deployment. In this case, the quality of the agreement point achieved by the protocol should be interpreted as a more general approximation: an approximation of the desired voting aggregation by the protocol that can be deployed in practice. An intriguing facet of this idea is that it can be viewed as an attempt to approximate vote aggregation that is based on {\em private} evaluations of (rewards from) various options by voters. This is in sharp contrast to the standard approach that focuses on performance measures based on {\em disclosed} information. E.g., using the standard approach, we define $PoA+$ in terms of candidate scoring, i.e., information that is explicitly disclosed by voters.

As for our developed framework, we hope to accommodate real-life decisions instead of games. Hopefully, voters will be able to set their preferences for real decisions that they care about, and then use the system to find a consensus. 
Lastly, it would be interesting to investigate and develop bots who change their strategies depending on whether other voters are humans or bots. 



\newpage
\bibliographystyle{spbasic}
\bibliography{ref}

\newpage
\appendix
\section{Proofs of Theorems, Lemmata and Formal Statements}

\textbf{Theorem 1}
\textit{In any $CUD(\calF^{IMaj}_\sigma, \tau)$ with $\sigma\in(\frac{n}{2},n]$ and consistent utility functions, either the default alternative is set at $t=0$, or a valid alternative becomes the winner at time $t\in[0:\tau]$.}

\begin{proof} Proof of Theorem \ref{thm:cuds_stop}.

 First, we prove by contradiction that if there are possible winners at the beginning of the process, then the process converges with some valid alternative as a winner. This implies that at stage $\tau$ the set of possible winners is not empty, nor is it empty at the final stage $0$, since there is a winner.
 Suppose that it does not hold: the process does not converge, even if there were possible winners at the beginning. In other words, at the final step $t=0$ the set of possible winners is empty, although at the initial step $t=\tau$ it is not empty.
This can happen if, at some time step between the beginning at $t=\tau$ and the end at $t=0$, the set of possible winners becomes empty (including the time step $t=0$). Consider the time step $\tau'$, such that for all $\tau''>\tau'$ the set of possible winners is not empty, while for $\tau'$ it is empty.

Consider the preceding iteration of our game protocol, i.e., the time step $\tau'+1$. There are two possible scenarios: (i) no voter changes the vote, and (ii) some voter changes the vote.

First, we assume that at the time step $\tau'+1$ there are no voters who wish to change their votes. Since at $\tau'+1$ the set of possible winners is not empty, for all the voters the utility is not zero under the current strategy. Since their utility at time $\tau'+1$ reflects the outcome of time $\tau'$, the set of possible winners could not become empty at time step $\tau'$.

Second, if the second scenario has occurred, at the time step $\tau'+1$ there are some voters who do change their vote. The only reason for a voter to change her vote is to improve her utility by switching away from an alternative that is either no longer in the set of possible winners, or will be removed from the set at the next time-slice. Such a switch, however, necessitates that the voter's newly chosen alternative remains in the set of possible winners at the next time-slice. Otherwise, the voter's utility from this new choice would not be positive, and no decision switch would occur. In particular, this implies that the set of possible winners is not empty. The obtained contradiction proves that if the set of possible winners is not empty at $\tau$, it cannot become empty at $0$, therefore there is a valid alternative that is  declared as a winner.

What remains is to show consistency of our protocol, i.e., that if at the time step $\tau$ the set of possible winners is empty, it must be empty at the time step $0$ as well.
 Denote $c \in C^+$ an alternative that at time step $\tau$ has the maximum score $s_c$ among all other alternatives. Since the set of possible winners is empty at the time step $\tau$, it implies that $\sigma-s_c \geq \tau+1$. Even if at each step until the deadline there will be a voter that changes her vote for this alternative $c$, at time $0$ this alternative will have obtained only $\tau$ more votes. 
 As a result, the score of the alternative $c$ at time step $0$ can be at most $s_c+\tau$, which would imply that $\sigma-s_c-\tau\geq 1$ and $c\not\in\widehat{W}(\vecs,0)$. As the score of $c$ is maximal possible at time $t=\tau$, the same holds for all other alternatives. Hence, no alternative is a member of $\widehat{W}(\vecs,0)$---it is empty.
 \end{proof}

\noindent
\textbf{Lemma 1}
\textit{Let $c \notin \widehat{W}$ at step $t$, then $c \notin \widehat{W}$ at any step $t'<t$.
}

\begin{proof} Proof of Lemma \ref{lemma:not_in_PW_before}.

  Given that $s_c^t<\sigma-t$ 
  and at each step the candidate $c$ can get at most one vote, then
  $s_c^{t'}<s_c^t+(t-t')<\sigma-t+(t-t')=\sigma-t'$ 
  for any $t'<t$. Therefore, $c \notin \widehat{W}$ at step $t'$.
\end{proof}

\noindent
\textbf{Lemma 2}
\textit{If $s_c^{t+1}<s_c^t$, then $c \in \widehat{W}(s^t,t)$.}

\begin{proof} Proof of Lemma \ref{lemma:if_1_vote_more}.

Given that $s_c^{t+1}<s_c^t$, there exists a voter, $i$, who changed her vote in favor of $c$ at time step $t$. Let $c'$ denote a candidate that she voted for at the preceding time step $t+1$, and let us assume that the lemma's conclusion does not hold. That is, let us assume that $c \notin \widehat{W}(s^t,t)$. 

Notice that, except $c$, no candidate increased his score from the time step $t+1$ to the time step $t$. This is because only one voter was given the chance to change her vote, and thus $s^t=s^{t+1}-c'+c$. 
As a result, $\widehat{W}(s^{t+1},t)\supseteq \widehat{W}(s^t,t)$. Therefore, either (but not both) of the following holds:
\begin{itemize}
\item $a_i(\topc_i(\widehat{W}(s^{t+1},t)), \topc_i(\widehat{W}(s^t,t)))$
\item $\topc_i(\widehat{W}(s^{t+1},t))= \topc_i(\widehat{W}(s^t,t))$.
\end{itemize}


A voter changes her vote only if it increases her utility, and conditions above indicate that a lazy voter's utility (Definition~\ref{def:lazy}) will not change between $s^{t+1}$ and $s^t$. Thus, voter $i$ can not be a lazy voter. 

Now, combining the fact that $s^t=s^{t+1}-c'+c$ with our attempt to assume that $c \notin \widehat{W}(s^t,t)$, we conclude that no candidate in the set $\widehat{W}(s^t,t)$ has a score higher that he has in $s^{t+1}$. More formally, $\forall \widehat{c}\in\widehat{W}(s^t,t),\ \ s^t_{\widehat{c}}\leq s^{t+1}_{\widehat{c}}$. As a result, switching from $c'$ to $c$ would not have been the preferred move of a proactive voter (Definition~\ref{def:act}), as it does not change the utility of the set of possible outcomes. 

We conclude that if the set of voters consists of lazy and/or proactive voters, then the assumption $c\notin \widehat{W}(s^t,t)$ leads to a contradiction of no voter having an incentive to change her vote. 
Thus, $c\in \widehat{W}(s^t,t)$ must hold.
\end{proof}

\noindent
\textbf{Lemma 3}
\textit{If a voter $j$ at the time step $t$ votes for candidate $c \in \widehat{W}(s^t,t)$, then  $c=\topc_i(\widehat{W}(s^t,t))$.}

\begin{proof} Proof of Lemma \ref{lemma:vote_for_top}.

First notice that all voters initially vote for their top choice, thus the lemma's conclusion holds for $t=\tau$. However, let us assume, to the contrary, that the lemma does not hold in general. In particular, it would imply that there is a time step $t$ such that for any $t'>t$ (i.e., preceding steps) the statement of the lemma is fulfilled, and at step $t$ there is a voter $j$ such that: (i) she votes for $c\in\widehat{W}(s^t,t)$; (ii) there is a candidate $\pord{c'}{c}{j}$ in the set $\widehat{W}(s^t,t)$. 

Lemma \ref{lemma:not_in_PW_before} implies that $\widehat{W}(s^t,t) \subseteq \widehat{W}(s^{t+1},t+1).$ If voter $j$ did not change her vote at time $t$, then by maximality of $t$ holds $c=\topc_i\in\widehat{W}(s^{t+1},t+1)$ and it implies that $c=\topc_i\in\widehat{W}(s^t,t)$, which we assumed not to hold. Thus, $j$ must have changed her vote at time step $t$. However, since there is $\pord{c'}{c}{j}$ in the set $\widehat{W}(s^t,t)$, both lazy and proactive consistent utility (see Definitions~\ref{def:lazy},~\ref{def:act}) from $c$ is lower than it is from $c'$ at time step $t$, which contradicts optimality of choice in voting decisions (see Game Protocol~\ref{cud_game}, line 8). 
\end{proof}

\noindent
\textbf{Lemma 4}
\textit{If there is a voter that changes her ballot at time $t$ from voting for $c$ to voting for $c'$, then $s_c^{t+1} \leq s_{c'}^{t+1}$.}

\begin{proof} Proof of Lemma \ref{lemma:if_change_vote_score_changes}.

Let us assume the opposite to the lemma's conclusion, that is $s_{c}^{t+1}>s_{c'}^{t+1}$ and consequently, $s_c^{t+1} \geq s_{c'}^{t+1}+1=s_{c'}^t$.  As a result, according to Lemma \ref{lemma:if_1_vote_more}, $c' \in \widehat{W}(s^t,t)$. However, since $s_{c'}^t$ is sufficient to become a possible winner at time $t$ and $s_c^{t+1}\geq s_{c'}^t$, holds that $c \in \widehat{W} (s^{t+1},t)$. Note that, according to Lemma~\ref{lemma:not_in_PW_before}, the above implies that $c,c' \in \widehat{W}(s^{t+1},t+1)$.

Now, Lemma \ref{lemma:vote_for_top} implies that if a voter votes for a candidate in the set $\widehat{W}(s^{t+1},t+1)$, then she prefers this candidate over all other possible winners. Therefore, similarly to the proof of Lemma~\ref{lemma:if_1_vote_more}, neither a lazy nor a proactive voter would change their votes. The obtained contradiction proves the lemma. 
\end{proof}

\noindent
\textbf{Corollary 1} 
\textit{Let $\veca=(a_1,\ldots,a_n)$ be the truthful profile, let $\tau$ be the deadline time, and let $\vecb$ be the ballot profile induced by $\veca$, i.e., $b_i=\topc_i(C)$. CUD stops with some $w\in C^+$  if and only if there is an alternative $c\in C^+$ so that $\score_c(\vecb)\geq \sigma-\tau$.}

\begin{proof} Proof of corollary \ref{thm:finer_bound}.
  
First, let us assume the right-hand side of the ``if and only if'' statement, and show that CUD stops with a non-default alternative.  The condition $\score_c(\vecb)\geq \sigma-\tau$ implies that $c$ is a possible winner by the definition of $\widehat{W}(\vecb,\tau)$. Thus, at time step $t=\tau$ the set of possible winners  contains at least one alternative. As a result, by Theorem \ref{thm:cuds_stop}, the process converges with some (non-default) alternative chosen as the winner.

Now, let us deal with the opposite direction of the Theorem's implication. Let CUD stop with some $w \in C^+$. Then this candidate $w$ achieved $\sigma$ votes in no more than $\tau$ steps. At each step he could  get 1 vote at most, that is, he achieved no more than $\tau$ votes. Thus, for the initial score of $w$ it must hold that $\score_w(\vecb)\geq \sigma-\tau$.
  \end{proof}
  
\noindent
\textbf{Theorem 2}
\textit{Let $\veca=(a_1,\ldots,a_n)$ be the truthful profile, let $\tau$ be the deadline time, and let $\vecb$ be the ballot profile induced by $\veca$, i.e., $b_i=\topc_i(C)$). If there is an alternative $c\in C^+$ so that $\score_c(\vecb)\geq\max\left\{\left\lfloor\frac{n}{2}\right\rfloor+1,\sigma-\tau\right\}$ then CUD terminates with $c$ as the winner.}

\begin{proof} Proof of Theorem \ref{thm:single_winner}.

Note that at any step such that all voters whose top-choice is candidate $c$, vote for $c$, for any other candidate it is true that $s_{c'}^t \leq n- s_c^t \leq \left\lfloor\frac{n}{2}\right\rfloor$ and $s_{c'}^t \leq s_c^t -1$.

Thus, if CUD terminates with a winner other than $c$, it implies that $c$ loses some votes of those whose top-choice is $c$. 

Consider $t$ such that for every $t' \geq t$ all voters whose top-choice is $c$ vote for $c$ and at step $t$ one of them changes her vote to $c'$. Thus, given that no one of them has changed their vote before, $s_{c'}^t<s_c^t$ which contradicts Lemma \ref{lemma:if_change_vote_score_changes}. That is, there is no such $t$. 

Therefore, candidate $c$ retains the same number of votes, $\score_c^\tau$, until step 0, which implies that at that last step he has more votes than any other candidate. 

%
%
%
%
\end{proof}

\noindent
\textbf{Theorem 3}
\textit{Let $\veca$ be the truthful profile of voters participating in a CUD. Assuming that it is well-defined for the CUD instance, the following bounds can be placed on the additive Price of Anarchy, $PoA^+$, depending on the ratio of the deadline timeout $\tau$ and the number of voters $n$:
\begin{enumerate}
\item If $\tau \leq \sigma - \left\lfloor\frac{n}{2}\right\rfloor$, then \begin{equation}PoA^+(\veca)=0.\tag{\ref{eq:poa_bound_1}}\end{equation}
\item If $\sigma-\left\lfloor\frac{n}{2}\right\rfloor<\tau<\sigma$, then \begin{equation}PoA^+(\veca) \leq \left\lfloor\frac{n}{2}\right\rfloor+\tau -\sigma\tag{\ref{eq:poa_bound_2}}\end{equation}
\item If $\tau\geq \sigma$, then
\begin{equation}PoA^+(\veca) \leq \left\lfloor\frac{n}{2}\right\rfloor-1.\tag{\ref{eq:poa_bound_3}}\end{equation}
\end{enumerate}}

\begin{proof} Proof of Theorem \ref{thm:poa_bounds}.

{\bf Case when  $\tau \leq \sigma - \big\lfloor\frac{n}{2}\big\rfloor$.}

 Note that, even if at each step every voter would change her vote in favour of the same alternative $c$, this alternative $c$ would get no more than $\tau$ points.

 \[s_c^{\tau}+\tau \geq \sigma \Leftrightarrow s_c^{\tau} \geq \sigma - \tau \geq \big\lfloor\frac{n}{2}\big\rfloor\]

 Note that, there can be at most two alternatives with score $\big\lfloor\frac{n}{2}\big\rfloor$.
 If $c$ is the only one, then he is the winner, hence $PoA^+(\veca)=0$.
 Suppose there are two such alternatives: $w$ and $c$. If they have equal scores at $\tau$: $s_w^\tau=s_c^\tau=\big\lfloor\frac{n}{2}\big\rfloor$, then,  whoever wins, $PoA^+(\veca)=0$.
 Another possibility is that there are two such alternatives, $w$ and $c$, such that $w$ has more points, i.e., $s_w^\tau=s_c^\tau +1$, and  all other alternatives have 0  points. Alternative $c$ would win only if every supporter of $w$ would change to $c$, which would take all $\tau$ stages. But, from those who initially voted for $w$, no voter would change her vote to $c$, since they are better off by keeping their votes for $w$. Hence, $w$ will win, and consequently,  $PoA^+(\veca)=0$.

{\bf Case when $\sigma-\left\lfloor\frac{n}{2}\right\rfloor<\tau<\sigma$}

Let $\omega$ denote the plurality winner at the time step $t=\tau$ and $c$ denote the winner at the time step $t=0$, and let us assume the contrary to the Theorem, i.e., $PoA^+(\veca)>\left\lfloor\frac{n}{2}\right\rfloor+\tau-\sigma$. In particular, it would mean that $\omega\neq c$ and 
\begin{equation}
s_{\omega}^\tau-s_c^\tau>\left\lfloor\frac{n}{2}\right\rfloor+\tau - \sigma\label{eq:theorem3.2}
\end{equation}
In $\tau$ steps candidate $c$ obtains at most $\tau$ votes and becomes a winner, that is $s_c^\tau+\tau \geq \sigma$. Therefore, $s_c^\tau \geq \sigma-\tau$. Combining this with Equation~\ref{eq:theorem3.2}, we obtain:

\[s_{\omega}^\tau>\left\lfloor\frac{n}{2}\right\rfloor+\tau - \sigma +s_c^\tau \geq \left\lfloor\frac{n}{2}\right\rfloor+\tau-\sigma+\sigma-\tau=\left\lfloor\frac{n}{2}\right\rfloor\]

Thus, $s_{\omega}^\tau>\left\lfloor\frac{n}{2}\right\rfloor$ and, consequently, $s_{\omega}^\tau \geq \left\lfloor\frac{n}{2}\right\rfloor+1$. However, according to Theorem~\ref{thm:single_winner}, this means that $\omega$ is the declared winner of the CUD at time step $t=0$. Which contradicts key part of our assumption: $c\neq\omega$. We thus must conclude that  $PoA^+(\veca) \leq \left\lfloor\frac{n}{2}\right\rfloor+\tau -\sigma$, as is per the Theorem.
 
 {\bf Case $\tau\geq \sigma$.}
 
 Notice again that if $PoA^+ \neq 0$ then the winner at the time step $t=\tau$ (denoted by $\omega$) and the winner at the time step $t=0$ (denoted by $c$) must be different. Now, Theorem \ref{thm:single_winner} implies that at the time step $t=\tau$ the truthful profile winner, $\omega$, can have at most $\left\lfloor\frac{n}{2}\right\rfloor$ votes, otherwise it must be the winner at the time step $t=0$ as well. 

At the same time, it must be that $s_c^\tau \geq 1$. Otherwise, no voter would be able to switch to $c$ at any time, as there will be at least one other candidate with a higher score than $c$ (the current winner) and, thus, Lemma~\ref{lemma:if_change_vote_score_changes} would prevent the switch. Since possible winners never lose votes, but only gain them, which means that the score of $c$ will never drop to zero either. Thus, $PoA^+(a) \leq \left\lfloor\frac{n}{2}\right\rfloor-1$, as required.
 \end{proof} 

\noindent
\textbf{Lemma 5}
\textit{The last two bounds in Theorem~\ref{thm:poa_bounds} are tight. For all $\tau$ and $n$ that satisfy the conditions of Equations~\ref{eq:poa_bound_2} and~\ref{eq:poa_bound_3}, there exists a truthful profile $\veca$ such that the corresponding bound holds as an equality.
}
 
\begin{proof} Proof of Lemma \ref{lemma:poa_bounds_tight}.

Table~\ref{tab:lemma5blocks} provides an example of a voting profile that proves that the bounds in Case 2 ($\sigma-\left\lfloor\frac{n}{2}\right\rfloor<\tau<\sigma$) of Theorem \ref{thm:poa_bounds} are tight. All voters are grouped into 3 Blocks. Voters in Block-1 prefer the candidate $c$ over $c_1$ and over all other candidates; voters in Block-2 prefer candidate $\omega$ over $c$ and over all other candidates; finally, each voter in Block-3 prefers some distinct candidate (but not $\omega$) over $c$ and over all other candidates. We assume that there are $\sigma-\tau$ voters in Block-1, $\left\lfloor\frac{n}{2}\right\rfloor$ voters in Block-2, and the rest of the voters are in Block-3.  It is assumed that $\sigma-\tau\geq 2$, so, there are at least 2 voters in Block-1.

\begin{table}[H]\caption{Proof of Lemma \ref{lemma:poa_bounds_tight}: voters' preferences}
 \label{tab:lemma5blocks}
\begin{center}
 
 \begin{tabular}{|c c c c|cc c c| c c c c|}
\hline
\multicolumn{4}{|c|}{Block 1 }& \multicolumn{4}{|c|}{Block 2 } & \multicolumn{4}{|c|}{Block 3 } \\
\multicolumn{4}{|c|}{($\sigma-\tau$ voters)}& \multicolumn{4}{|c|}{($\left\lfloor\frac{n}{2}\right\rfloor$ voters)} & \multicolumn{4}{|c|}{(The rest of the voters)} \\ \hline
$c$ & $c$ &$...$& $c$ & $w$ &$w$&$...$&$w$ & $c_1$&$c_2$&$...$&$c_k$\\
$c_1$& $c_1$ &$...$ & $c_1$ & $c$ & $c$ &$...$& $c$ &$c$ & $c$ &$...$& $c$ \\
\multicolumn{4}{|c|}{\scriptsize{All other candidates in any order}}&\multicolumn{4}{|c|}{\scriptsize{All other candidates in any order}}&\multicolumn{4}{|c|}{\scriptsize{All other candidates in any order}}\\
\hline
\end{tabular}

\end{center}
\end{table}

Notice that we can indeed construct such a preference profile given that $\sigma- \left\lfloor\frac{n}{2}\right\rfloor<\tau$, that is, $\sigma - \tau < \left\lfloor\frac{n}{2}\right\rfloor$.
In particular, the size of Block-3 is 
$k=n-\left\lfloor\frac{n}{2}\right\rfloor-\sigma + \tau$, and, therefore, every candidate from $\{c_1,c_2,...c_k\}$ appears as a top-choice only once. 

We assume that $\sigma - \tau \geq 2$ and $n$ is odd.  Then, 
$\widehat{W}(s^{\tau},\tau)=\{c,w\}$. Thus, at any time step prior to $\sigma - \left\lfloor\frac{n}{2}\right\rfloor$ voters from the Block-3 will want to change their vote to $c$, since they currently vote for a candidate outside the set of possible winner $\widehat{W}$ and for all of these voters $c\succ w$. No candidate from Block-1 or Block-2 will want to change their votes during that time. 


As a result, for $t=\left\lfloor\frac{n}{2}\right\rfloor-\sigma + \tau$ holds $s_c^{t}=s_w^{t}=\left\lfloor\frac{n}{2}\right\rfloor$. After this step all voters from Block 1 and all except one from Block 3 vote for $c$, and all the voters from Block 2 vote for $w$.  At the next step, $t-1$, all the voters will want to change their vote since the set of possible winners $\widehat{W}(s^{t-1},t-1)$ can only contain candidates with $\left\lfloor\frac{n}{2}\right\rfloor+1$ votes. 

Ties among voters who wish to change their vote are broken randomly, and with probability $\frac{\left\lfloor\frac{n}{2}\right\rfloor+1}{n}$ a voter will be chosen who will change her vote to $c$. This will make $c$ the only possible winner, which, more formally, means that $\widehat{W}(s^{\left\lfloor\frac{n}{2}\right\rfloor-\sigma+\tau-1},$ \\$\left\lfloor\frac{n}{2}\right\rfloor-\sigma+\tau -1)=\{c\}$. Hence, $c$ will be the winner of the entire election process, the final winner. Given that $c$ and $w$ are the only possible winners, and $w$ is a Plurality winner at time $\tau$, $PoA^+=\left\lfloor\frac{n}{2}\right\rfloor-\sigma+\tau$.

Table \ref{tab:lemma5blocks2} provides an example of a voting profile that proves that the bounds in the Case 3 ($\tau\geq \sigma$) of Theorem \ref{thm:poa_bounds} are tight. 

Once again, all voters are grouped into three blocks. Voters of Block-1 prefer candidate $\omega$ over $c$ and over all other candidates; Block-2 voters prefer candidate  $c$ over $c_1$ and over all other candidates; each voter in Block-3 prefers some distinct candidate (but not $\omega$) over candidate $c$ and other candidates. Let there be $\left\lfloor\frac{n}{2}\right\rfloor$ voters in Block-1, a single voter in Block-2, and the rest of the voters grouped into Block-3. We will assume that $n$ is odd, so that the number of voters in Block-3 is $k=\left\lfloor\frac{n}{2}\right\rfloor$.

\begin{table}[H]\caption{Proof of Lemma \ref{lemma:poa_bounds_tight}: voters' preferences -- point of ref}
 \label{tab:lemma5blocks2}
 \begin{center}
 \begin{tabular}{|c c c c|c| c c c c|}
\hline
\multicolumn{4}{|c|}{Block 1}& \multicolumn{1}{|c|}{Block 2} & \multicolumn{4}{|c|}{Block 3} \\
\multicolumn{4}{|c|}{($\left\lfloor\frac{n}{2}\right\rfloor$ voters)}& \multicolumn{1}{|c|}{(1 voter) } & \multicolumn{4}{|c|}{(The rest of the voters)} \\ \hline
$w$ & $w$ &$...$& $w$ & $c$ & $c_1$&$c_2$&$...$&$c_k$\\
$c$& $c$ &$...$ & $c$ & $c_1$ &  $c$&$c$ &$...$& $c$ \\
\hline
\end{tabular}
 \end{center}
\end{table}

Now, notice that at the time step $t=\sigma-1$ voters from Block-2 and Block-3 will want to change their vote. In particular, voters from Block-3 will want to change their votes to $c$, and, because ties among willing voters are broken uniformly at random, with probability $\frac{\left\lfloor\frac{n}{2}\right\rfloor}{\left\lfloor\frac{n}{2}\right\rfloor+1}$ a voter from Block-3 will be granted the opportunity and change her vote in favour of $c$. Thus, $\widehat{W}(s^{\tau-1},\tau-1)=\{w,c\}$. Analogously to the previously investigated profile, both $w$ and $c$ can be the final winner. Hence, $PoA^+\geq \left\lfloor\frac{n}{2}\right\rfloor-1$. Furthermore, combining this conclusion with an application of Theorem~\ref{thm:poa_bounds} to the constructed profile, we obtain that $PoA^+= \left\lfloor\frac{n}{2}\right\rfloor-1$ for the constructed profile. This yields the tightness of Theorem~\ref{thm:poa_bounds} as required.

\end{proof}

\section{Remarks on Condorcet Winners}\label{appendix:condorcet}
In addition to the proofs of our main results, we would like to venture some additional thoughts on refinements of our protocol. Specifically, we would like to make a few notes about its relationship to Condorcet winners.

\begin{lemma}
Protocol\ref{cud_game} is not Condorcet-consistent.
\end{lemma}
\begin{proof}
Consider a preference profile, where an alternative $c\in C$ is the second choice for all voters, and the top choices divide equally among ${a_1,a_2,a_3}\subseteq C\setminus\{c\}$. Furthermore, assume that $c$ is not the default alternative.  

It is easy to see that $c$ is a Condorcet winner. However, since the initial score of $c$ is zero and, as a consequence of Lemma~\ref{lemma:if_change_vote_score_changes}, no voter will change her ballot in favour of $c$, and it will persistently remain outside of the set of possible winners. Hence, Protocol~\ref{cud_game} will either result in a default or any other alternative in $C\setminus\{c\}$. That is, Protocol~\ref{cud_game} is not Condorcet-consistent. 
\end{proof}

However, it is marginally easier for a Condorcet winner to survive our Protocol and become its winner, as the following augmentation of Theorem~\ref{thm:single_winner} states.
\begin{conj}
Let $\veca=(a_1,\ldots,a_n)$ be the truthful profile, let $\tau$ be the deadline time, and let $\vecb$ be the ballot profile induced by $\veca$, i.e., $b_i=\topc_i(C)$). Let $c\in C^+$ be a Condorcet winner under the profile $\veca$. If in addition holds that  $\score_c(\vecb)\geq\max\left\{\left\lceil\frac{n}{2}\right\rceil-1,\sigma-\tau\right\}$ then CUD terminates with $c$ as the winner.
\end{conj}

\section{Full simulation results }\label{appendix:uniform}

 \begin{table} 
\caption{Number of vote changes: the average  and standard deviation of the number of vote changes required to reach a consensus.}
\label{appendixtable:2}
\begin{small}
\begin{center}
   \begin{tabular}{| c| c c| c c| }
    \hline
    & \multicolumn{2}{c |}{\textbf{10 voters}} & \multicolumn{2}{c |}{\textbf{20 voters}} \\
     \textbf{Data set} & \textbf{lazy} & \textbf{proactive} & \textbf{lazy} & \textbf{proactive}  \\
     \hline
    \textbf{Uniform5} & $5.06  \pm  2.95$ & $5.15  \pm 3.02  $ & $10.74  \pm 6.02  $ & $10.83 \pm 6.1   $   \\
    \textbf{Uniform6} &  $5.03 \pm 3.64$ & $5.09 \pm 3.7$ & $10.11 \pm 7.71$ & $10.15 \pm 7.75$ \\
    \textbf{Uniform7} & $5 \pm 3.74$ &   $5.05 \pm 3.79$ & $9.98 \pm 8.05$ & $10.01 \pm 8.07$\\
    \textbf{Uniform8} &  $5 \pm 3.92$ & $5.05 \pm 3.97$  & $9.89 \pm 8.13$ & $9.92 \pm 8.17$ \\
    \textbf{Courses 2003} & $4.98 \pm 2.45 $   & $6.46  \pm 3.54 $  &$ 11.31 \pm 5.45  $ & $13.09  \pm 6.72  $  \\
    \textbf{Courses 2004} & $3.55 \pm 2.05  $ & $5.73  \pm 3.78 $ & $7.5  \pm 4.15  $&$ 11.19  \pm 7.04 $  \\
    \textbf{Sushi} & $5.71 \pm 2.9 $ & $ 6.75 \pm 3.7 $ &$ 11.11  \pm 6.32 $ & $12.42 \pm 7.31 $ \\
    \textbf{T-shirts} & $ 5.66 \pm 2.5  $ & $ 6.91 \pm 3.45  $ & $12.38 \pm 5.97  $ & $13.76 \pm 6.93   $   \\
    \hline
    \end{tabular}
    
    \begin{tabular}{| c| c c|   }
    \hline
    & \multicolumn{2}{c |}{\textbf{30 voters}}   \\
     \textbf{Data set} & \textbf{lazy} & \textbf{proactive}   \\
     \hline
    \textbf{Uniform5}   & $19.46 \pm 7.6     $& $19.49 \pm $ 7.63 \\
    \textbf{Uniform6}   & $18.55 \pm 10.1$ & $18.64 \pm 10.17$  \\
    \textbf{Uniform7}  & $17.39 \pm 11.15$ & $17.43 \pm 11.18$   \\
    \textbf{Uniform8}   & $17.22 \pm 11.54$ & $17.24 \pm 11.55$  \\
    \textbf{Courses 2003}   &$ 18.7 \pm 6.7 $    &$ 21.17 \pm 8.22$\\
    \textbf{Courses 2004} &  $11.98 \pm  6.37 $ & $ 17.13 \pm 10.25$\\
    \textbf{Sushi}  & $18.61 \pm 8.14 $ & $20.71  \pm 9.52 $\\
    \textbf{T-shirts}   & $18.74 \pm 9.43    $ & $ 20.43 \pm 10.6 $ \\
    \hline
    \end{tabular}
\end{center}
\end{small}
\end{table}

\begin{table}
\caption{The average and standard deviations of the upper bounds for the additive price of anarchy for 10 voters for different initial $\tau$ (time until the deadline). The $PoA^+$ is identical for both lazy and proactive voters. }
\label{appendixtable:3}
\begin{small}
\begin{center}

\begin{tabular}{|L|c|c|c|c|c|}\hline
      \textbf{Time until the deadline}&\textbf{Uniform5}&\textbf{Uniform6}&\textbf{Uniform7}&\textbf{Uniform8} \\\hline
      \textbf{2}&$0 \pm 0$  & $0 \pm 0$  &$0 \pm 0$  &$0 \pm 0$ \\
\textbf{3}&$0 \pm 0$  & $0 \pm 0$  &$0 \pm 0$  &$0 \pm 0$ \\
\textbf{4}&$0 \pm 0$  & $0 \pm 0$  &$0 \pm 0$  &$0 \pm 0$ \\
\textbf{5}&$0 \pm 0$  & $0 \pm 0$  &$0 \pm 0$  &$0 \pm 0$\\ 
\textbf{6}&$0 \pm 0$  & $0 \pm 0$  &$0 \pm 0$  &$0 \pm 0$ \\
\textbf{7}&$0.1 \pm 0.3$  & $0.15 \pm 0.37$  &$0 \pm 0$  &$0.1 \pm 0.3$ \\
\textbf{8}&$0.85 \pm 0.74$  & $0.85 \pm 0.67$  &$0.75 \pm 0.63$  &$0.65 \pm 0.66$ \\
\textbf{9}&$1 \pm 0.91$  & $1.2 \pm 0.77$  &$1.1 \pm 0.78$  &$1.5 \pm 0.82$ \\
\textbf{10}&$1 \pm 0.91$  & $1.2 \pm 0.77$  &$1.1 \pm 0.78$  &$1.5 \pm 0.82$ \\
\textbf{11}&$1 \pm 0.91$  & $1.2 \pm 0.77$  &$1.1 \pm 0.78$  &$1.5 \pm 0.82$ \\

         \hline
    \end{tabular}

  \begin{tabular}{|L|c|c|c|c|c|}\hline
      \textbf{Time until the deadline}&\textbf{Courses 2003}&\textbf{Courses 2004}&\textbf{Sushi} &\textbf{T-Shirts}\\\hline
    \textbf{2}   & $0 \pm 0$  &$0 \pm 0$  &$0 \pm 0$  &$0 \pm 0$  \\
    \textbf{3}    & $0 \pm 0$  &$0 \pm 0$  &$0 \pm 0$  &$0 \pm 0$   \\
    \textbf{4}    & $0 \pm 0$  &$0 \pm 0$  &$0 \pm 0$  &$0 \pm 0$    \\
    \textbf{5}    & $0 \pm 0$  &$0 \pm 0$  &$0 \pm 0$  &$0 \pm 0$    \\
    \textbf{6}   & $0 \pm 0$  &$0 \pm 0$  &$0 \pm 0$  &$0 \pm 0$  \\
    \textbf{7}    & $0.15 \pm 0.36$  &$0.05 \pm 0.22$  &$0 \pm 0$  &$0.15 \pm 0.36$  \\
    \textbf{8}    & $0.4 \pm 0.59$  &$0.05 \pm 0.22$  &$0.15 \pm 0.48$  &$0.55 \pm 0.81$  \\
    \textbf{9}    & $0.7 \pm 1.02$  &$0.05 \pm 0.22$  &$0.45 \pm 1.04$  &$0.65 \pm 0.86$  \\
    \textbf{10}    & $0.7 \pm 1.02$  &$0.05 \pm 0.22$  &$0.45 \pm 1.04$  &$0.65 \pm 0.86$  \\
    \textbf{11}    & $0.7 \pm 1.02$  &$0.05 \pm 0.22$  &$0.45 \pm 1.04$  &$0.65 \pm 0.86$  \\
    \hline
    \end{tabular}%

\end{center}
\end{small}
\end{table}

\begin{table}
\caption{The average and standard deviations of the upper bounds for the additive price of anarchy for 20 voters for different initial $\tau$ (time until the deadline). The $PoA^+$ is identical for both lazy and proactive voters. }
\label{appendixtable:4}
\begin{small}
\begin{center}

\begin{tabular}{|L|c|c|c|c|c|}\hline
      \textbf{Time until the deadline}&\textbf{Uniform5}&\textbf{Uniform6}&\textbf{Uniform7}&\textbf{Uniform8} \\\hline
         \textbf{2}&$0 \pm 0$  & $0 \pm 0$  &$0 \pm 0$  &$0 \pm 0$ \\
\textbf{3}&$0 \pm 0$  & $0 \pm 0$  &$0 \pm 0$  &$0 \pm 0$ \\
\textbf{4}&$0 \pm 0$  & $0 \pm 0$  &$0 \pm 0$  &$0 \pm 0$ \\
\textbf{5}&$0 \pm 0$  & $0 \pm 0$  &$0 \pm 0$  &$0 \pm 0$ \\
\textbf{6}&$0 \pm 0$  & $0 \pm 0$  &$0 \pm 0$  &$0 \pm 0$ \\
\textbf{7}&$0 \pm 0$  & $0 \pm 0$  &$0 \pm 0$  &$0 \pm 0$ \\
\textbf{8}&$0 \pm 0$  & $0 \pm 0$  &$0 \pm 0$  &$0 \pm 0$ \\
\textbf{9}&$0 \pm 0$  & $0 \pm 0$  &$0 \pm 0$  &$0 \pm 0$ \\
\textbf{10}&$0 \pm 0$  & $0 \pm 0$  &$0 \pm 0$  &$0 \pm 0$ \\
\textbf{11}&$0 \pm 0$  & $0 \pm 0$  &$0 \pm 0$  &$0 \pm 0$ \\
\textbf{12}&$0 \pm 0$  & $0 \pm 0$  &$0 \pm 0$  &$0 \pm 0$ \\
\textbf{13}&$0.1 \pm 0.3$  & $0 \pm 0$  &$0 \pm 0$  &$0 \pm 0$ \\
\textbf{14}&$0.1 \pm 0.3$  & $0 \pm 0$  &$0 \pm 0$  &$0 \pm 0$ \\
\textbf{15}&$0.9 \pm 1.06$  & $0.15 \pm 0.37$  &$0 \pm 0$  &$0 \pm 0$ \\
\textbf{16}&$1.75 \pm 1.19$  & $1.2 \pm 1.15$  &$0.75 \pm 0.84$  &$0.85 \pm 1.08$ \\
\textbf{17}&$2.05 \pm 1.34$  & $2.15 \pm 0.99$  &$1.9 \pm 1.01$  &$1.8 \pm 0.82$ \\
\textbf{18}&$2.05 \pm 1.34$  & $2.3 \pm 0.98$  &$2.6 \pm 1.17$  &$2.4 \pm 1.08$ \\
\textbf{19}&$2.05 \pm 1.34$  & $2.3 \pm 0.98$  &$2.6 \pm 1.17$  &$2.4 \pm 1.08$ \\
\textbf{20}&$2.05 \pm 1.34$  & $2.3 \pm 0.98$  &$2.6 \pm 1.17$  &$2.4 \pm 1.08$ \\
\textbf{21}&$2.05 \pm 1.34$  & $2.3 \pm 0.98$  &$2.6 \pm 1.17$  &$2.4 \pm 1.08$ \\

         \hline
    \end{tabular}

  \begin{tabular}{|L|c|c|c|c|c|}\hline
      \textbf{Time until the deadline}&\textbf{Courses 2003}&\textbf{Courses 2004}&\textbf{Sushi} &\textbf{T-Shirts}\\\hline
    \textbf{2}&$0 \pm 0$  & $0 \pm 0$  &$0 \pm 0$  &$0 \pm 0$ \\
\textbf{3}&$0 \pm 0$  & $0 \pm 0$  &$0 \pm 0$  &$0 \pm 0$ \\
\textbf{4}&$0 \pm 0$  & $0 \pm 0$  &$0 \pm 0$  &$0 \pm 0$ \\
\textbf{5}&$0 \pm 0$  & $0 \pm 0$  &$0 \pm 0$  &$0 \pm 0$ \\
\textbf{6}&$0 \pm 0$  & $0 \pm 0$  &$0 \pm 0$  &$0 \pm 0$ \\
\textbf{7}&$0 \pm 0$  & $0 \pm 0$  &$0 \pm 0$  &$0 \pm 0$ \\
\textbf{8}&$0 \pm 0$  & $0 \pm 0$  &$0 \pm 0$  &$0 \pm 0$ \\
\textbf{9}&$0 \pm 0$  & $0 \pm 0$  &$0 \pm 0$  &$0 \pm 0$ \\
\textbf{10}&$0 \pm 0$  & $0 \pm 0$  &$0 \pm 0$  &$0 \pm 0$ \\
\textbf{11}&$0 \pm 0$  & $0 \pm 0$  &$0 \pm 0$  &$0 \pm 0$ \\
\textbf{12}&$0.05 \pm 0.22$  & $0 \pm 0$  &$0 \pm 0$  &$0 \pm 0$ \\
\textbf{13}&$0.05 \pm 0.22$  & $0.2 \pm 0.52$  &$0 \pm 0$  &$0 \pm 0$ \\
\textbf{14}&$0.2 \pm 0.52$  & $0.25 \pm 0.55$  &$0.1 \pm 0.45$  &$0.05 \pm 0.22$ \\
\textbf{15}&$0.65 \pm 0.99$  & $0.25 \pm 0.55$  &$0.65 \pm 0.88$  &$0.3 \pm 0.57$ \\
\textbf{16}&$1.35 \pm 1.39$  & $0.25 \pm 0.55$  &$1.5 \pm 1.32$  &$0.95 \pm 0.76$ \\
\textbf{17}&$1.85 \pm 1.79$  & $0.25 \pm 0.55$  &$1.95 \pm 1.47$  &$1.75 \pm 1.21$ \\
\textbf{18}&$1.9 \pm 1.86$  & $0.25 \pm 0.55$  &$2.1 \pm 1.55$  &$2.45 \pm 1.64$ \\
\textbf{19}&$1.9 \pm 1.86$  & $0.25 \pm 0.55$  &$2.4 \pm 1.7$  &$2.45 \pm 1.64$ \\
\textbf{20}&$1.9 \pm 1.86$  & $0.25 \pm 0.55$  &$2.4 \pm 1.7$  &$2.45 \pm 1.64$ \\
\textbf{21}&$1.9 \pm 1.86$  & $0.25 \pm 0.55$  &$2.4 \pm 1.7$  &$2.45 \pm 1.64$ \\

    \hline
    \end{tabular}%

\end{center}
\end{small}
\end{table}

\begin{table}
\caption{The average and standard deviations of the upper bounds for the additive price of anarchy for 30 voters for different initial $\tau$ (time until the deadline). The $PoA^+$ is identical for both lazy and proactive voters. }
\label{appendixtable:5}
\begin{small}
\begin{center}

\begin{tabular}{|L|c|c|c|c|c|}\hline
      \textbf{Time until the deadline}&\textbf{Uniform5}&\textbf{Uniform6}&\textbf{Uniform7}&\textbf{Uniform8} \\\hline
          \textbf{2}&$0 \pm 0$  & $0 \pm 0$  &$0 \pm 0$  &$0 \pm 0$ \\
\textbf{3}&$0 \pm 0$  & $0 \pm 0$  &$0 \pm 0$  &$0 \pm 0$ \\
\textbf{4}&$0 \pm 0$  & $0 \pm 0$  &$0 \pm 0$  &$0 \pm 0$ \\
\textbf{5}&$0 \pm 0$  & $0 \pm 0$  &$0 \pm 0$  &$0 \pm 0$ \\
\textbf{6}&$0 \pm 0$  & $0 \pm 0$  &$0 \pm 0$  &$0 \pm 0$ \\
\textbf{7}&$0 \pm 0$  & $0 \pm 0$  &$0 \pm 0$  &$0 \pm 0$ \\
\textbf{8}&$0 \pm 0$  & $0 \pm 0$  &$0 \pm 0$  &$0 \pm 0$ \\
\textbf{9}&$0 \pm 0$  & $0 \pm 0$  &$0 \pm 0$  &$0 \pm 0$ \\
\textbf{10}&$0 \pm 0$  & $0 \pm 0$  &$0 \pm 0$  &$0 \pm 0$ \\
\textbf{11}&$0 \pm 0$  & $0 \pm 0$  &$0 \pm 0$  &$0 \pm 0$ \\
\textbf{12}&$0 \pm 0$  & $0 \pm 0$  &$0 \pm 0$  &$0 \pm 0$ \\
\textbf{13}&$0 \pm 0$  & $0 \pm 0$  &$0 \pm 0$  &$0 \pm 0$ \\
\textbf{14}&$0 \pm 0$  & $0 \pm 0$  &$0 \pm 0$  &$0 \pm 0$ \\
\textbf{15}&$0 \pm 0$  & $0 \pm 0$  &$0 \pm 0$  &$0 \pm 0$ \\
\textbf{16}&$0 \pm 0$  & $0 \pm 0$  &$0 \pm 0$  &$0 \pm 0$ \\
\textbf{17}&$0 \pm 0$  & $0 \pm 0$  &$0 \pm 0$  &$0 \pm 0$ \\
\textbf{18}&$0 \pm 0$  & $0 \pm 0$  &$0 \pm 0$  &$0 \pm 0$ \\
\textbf{19}&$0 \pm 0$  & $0 \pm 0$  &$0 \pm 0$  &$0 \pm 0$ \\
\textbf{20}&$0 \pm 0$  & $0 \pm 0$  &$0 \pm 0$  &$0 \pm 0$ \\
\textbf{21}&$0.15 \pm 0.36$  & $0 \pm 0$  &$0 \pm 0$  &$0 \pm 0$ \\
\textbf{22}&$0.2 \pm 0.41$  & $0 \pm 0$  &$0 \pm 0$  &$0 \pm 0$ \\
\textbf{23}&$1.15 \pm 1.21$  & $0.2 \pm 0.52$  &$0 \pm 0$  &$0.05 \pm 0.22$ \\
\textbf{24}&$2.25 \pm 1.35$  & $1.4 \pm 1.14$  &$0.55 \pm 0.81$  &$0.4 \pm 0.87$ \\
\textbf{25}&$2.7 \pm 1.47$  & $2.25 \pm 1.02$  &$1.55 \pm 1.13$  &$1.6 \pm 1.03$ \\
\textbf{26}&$2.7 \pm 1.47$  & $3.1 \pm 1.25$  &$2.5 \pm 1.26$  &$2.7 \pm 1.07$ \\
\textbf{27}&$2.7 \pm 1.47$  & $3.1 \pm 1.25$  &$2.9 \pm 1.72$  &$3.1 \pm 1.15$ \\
\textbf{28}&$2.7 \pm 1.47$  & $3.1 \pm 1.25$  &$2.9 \pm 1.72$  &$3.1 \pm 1.15$ \\
\textbf{29}&$2.7 \pm 1.47$  & $2.89 \pm 1.25$  &$2.9 \pm 1.72$  &$3.1 \pm 1.15$ \\
\textbf{30}&$2.7 \pm 1.47$  & $3.1 \pm 1.24$  &$2.9 \pm 1.72$  &$3.1 \pm 1.15$ \\
\textbf{31}&$2.7 \pm 1.47$  & $3.1 \pm 1.24$  &$2.9 \pm 1.72$  &$3.1 \pm 1.15$ \\

         \hline
    \end{tabular}

  \begin{tabular}{|L|c|c|c|c|c|}\hline
      \textbf{Time until the deadline}&\textbf{Courses 2003}&\textbf{Courses 2004}&\textbf{Sushi} &\textbf{T-Shirts}\\\hline
    \textbf{2}&$0 \pm 0$  & $0 \pm 0$  &$0 \pm 0$  &$0 \pm 0$ \\
\textbf{3}&$0 \pm 0$  & $0 \pm 0$  &$0 \pm 0$  &$0 \pm 0$ \\
\textbf{4}&$0 \pm 0$  & $0 \pm 0$  &$0 \pm 0$  &$0 \pm 0$ \\
\textbf{5}&$0 \pm 0$  & $0 \pm 0$  &$0 \pm 0$  &$0 \pm 0$ \\
\textbf{6}&$0 \pm 0$  & $0 \pm 0$  &$0 \pm 0$  &$0 \pm 0$ \\
\textbf{7}&$0 \pm 0$  & $0 \pm 0$  &$0 \pm 0$  &$0 \pm 0$ \\
\textbf{8}&$0 \pm 0$  & $0 \pm 0$  &$0 \pm 0$  &$0 \pm 0$ \\
\textbf{9}&$0 \pm 0$  & $0 \pm 0$  &$0 \pm 0$  &$0 \pm 0$ \\
\textbf{10}&$0 \pm 0$  & $0 \pm 0$  &$0 \pm 0$  &$0 \pm 0$ \\
\textbf{11}&$0 \pm 0$  & $0 \pm 0$  &$0 \pm 0$  &$0 \pm 0$ \\
\textbf{12}&$0 \pm 0$  & $0 \pm 0$  &$0 \pm 0$  &$0 \pm 0$ \\
\textbf{13}&$0 \pm 0$  & $0 \pm 0$  &$0 \pm 0$  &$0 \pm 0$ \\
\textbf{14}&$0 \pm 0$  & $0 \pm 0$  &$0 \pm 0$  &$0 \pm 0$ \\
\textbf{15}&$0 \pm 0$  & $0 \pm 0$  &$0 \pm 0$  &$0 \pm 0$ \\
\textbf{16}&$0 \pm 0$  & $0 \pm 0$  &$0 \pm 0$  &$0 \pm 0$ \\
\textbf{17}&$0 \pm 0$  & $0 \pm 0$  &$0 \pm 0$  &$0 \pm 0$ \\
\textbf{18}&$0 \pm 0$  & $0 \pm 0$  &$0 \pm 0$  &$0 \pm 0$ \\
\textbf{19}&$0.1 \pm 0.45$  & $0.15 \pm 0.67$  &$0 \pm 0$  &$0 \pm 0$ \\
\textbf{20}&$0.35 \pm 0.88$  & $0.45 \pm 1$  &$0 \pm 0$  &$0 \pm 0$ \\
\textbf{21}&$0.35 \pm 0.88$  & $0.45 \pm 1$  &$0.05 \pm 0.22$  &$0 \pm 0$ \\
\textbf{22}&$1.15 \pm 1.63$  & $0.45 \pm 1$  &$0.4 \pm 0.82$  &$0.05 \pm 0.22$ \\
\textbf{23}&$1.7 \pm 1.95$  & $0.45 \pm 1$  &$0.95 \pm 1.05$  &$0.25 \pm 0.44$ \\
\textbf{24}&$1.7 \pm 1.95$  & $0.45 \pm 1$  &$1.25 \pm 1.12$  &$0.7 \pm 0.66$ \\
\textbf{25}&$2.15 \pm 2.11$  & $0.45 \pm 1$  &$1.55 \pm 1.5$  &$1.75 \pm 1.21$ \\
\textbf{26}&$2.6 \pm 2.19$  & $0.45 \pm 1$  &$1.85 \pm 1.76$  &$2.75 \pm 1.8$ \\
\textbf{27}&$2.8 \pm 2.65$  & $0.45 \pm 1$  &$2.15 \pm 1.93$  &$3.2 \pm 1.94$ \\
\textbf{28}&$2.8 \pm 2.65$  & $0.45 \pm 1$  &$2.2 \pm 2.04$  &$3.35 \pm 1.95$ \\
\textbf{29}&$2.8 \pm 2.65$  & $0.45 \pm 1$  &$2.2 \pm 2.04$  &$3.35 \pm 1.95$ \\
\textbf{30}&$2.8 \pm 2.65$  & $0.45 \pm 1$  &$2.2 \pm 2.04$  &$3.35 \pm 1.95$ \\
\textbf{31}&$2.8 \pm 2.65$  & $0.45 \pm 1$  &$2.2 \pm 2.04$  &$3.35 \pm 1.95$ \\

    \hline
    \end{tabular}%

\end{center}
\end{small}
\end{table}

\end{document}